\documentclass[pra,twocolumn,preprintnumbers,amsmath,amssymb,superscriptaddress]{revtex4}

\usepackage{color}
\usepackage{graphicx}% Incl fgfrude figure files
\usepackage{dcolumn}% Align table columns on decimal point
\usepackage{bm}% bold math
\usepackage{hyperref}
\usepackage{enumitem}
\usepackage{balance}
\usepackage{comment}
\usepackage{cleveref}

\usepackage{amsthm}
\newtheorem{theorem}{Theorem}
\newtheorem{lemma}{Lemma}

\begin{document}

\title{Stroboscopic quantum nondemolition measurements for enhanced entanglement generation between atomic ensembles}
%\title{Measurement operator approach to stroboscopic quantum nondemolition measurement induced entanglement between atomic ensembles}

\author{Manish Chaudhary}
\affiliation{State Key Laboratory of Precision Spectroscopy, School of Physical and Material Sciences, East China Normal University, Shanghai 200062, China}
\affiliation{New York University Shanghai, 1555 Century Ave, Pudong, Shanghai 200122, China}

\author{Yuping Mao}
\affiliation{State Key Laboratory of Precision Spectroscopy, School of Physical and Material Sciences, East China Normal University, Shanghai 200062, China}
\affiliation{New York University Shanghai, 1555 Century Ave, Pudong, Shanghai 200122, China}

\author{Manikandan Kondappan}
\affiliation{State Key Laboratory of Precision Spectroscopy, School of Physical and Material Sciences, East China Normal University, Shanghai 200062, China}
\affiliation{New York University Shanghai, 1555 Century Ave, Pudong, Shanghai 200122, China}

\author{Amiel S. P. Paz}
\affiliation{New York University Shanghai, 1555 Century Ave, Pudong, Shanghai 200122, China} 

\author{Valentin Ivannikov}
\affiliation{New York University Shanghai, 1555 Century Ave, Pudong, Shanghai 200122, China} 
\affiliation{NYU-ECNU Institute of Physics at NYU Shanghai, 3663 Zhongshan Road North, Shanghai 200062, China}

\author{Tim Byrnes}
\email{tim.byrnes@nyu.edu}
\affiliation{New York University Shanghai, 1555 Century Ave, Pudong, Shanghai 200122, China}  
\affiliation{State Key Laboratory of Precision Spectroscopy, School of Physical and Material Sciences, East China Normal University, Shanghai 200062, China}
\affiliation{NYU-ECNU Institute of Physics at NYU Shanghai, 3663 Zhongshan Road North, Shanghai 200062, China}
\affiliation{National Institute of Informatics, 2-1-2 Hitotsubashi, Chiyoda-ku, Tokyo 101-8430, Japan}
\affiliation{Department of Physics, New York University, New York, NY 10003, USA}

\date{\today}% It is always \today, today,

\begin{abstract}
We develop a measurement operator formalism to  handle quantum nondemolition (QND) measurement induced entanglement generation between two atomic gases.  We first derive how the QND entangling scheme reduces to a positive operator valued measure (POVM), and consider its limiting case when it can be used to construct a projection operator that collapses the state to a total spin projection state.  We then analyze how a stroboscopic sequence of such projections made in the $x$ and $z$ basis evolves the initial wavefunction. Such a sequence of QND projections can enhance the entanglement between the atomic ensembles and makes the state converge towards a highly entangled state. We show several mathematical identities which greatly simplify the state evolution in the projection sequence, and allows one to derive the exact state in a highly efficient manner. Our formalism does not use the Holstein-Primakoff approximation as is conventionally done, and treats the spins of the atomic gases in an exact way.  
\end{abstract}

%\pacs{03.75.Dg, 37.25.+k, 03.75.Mn}% PACS, the Physics and Astronomy
\maketitle

\section{\label{sec1}Introduction}
%Entanglement and one atomic ensemble entanglement review
Entanglement is one of the fundamental phenomena observed in quantum mechanics \cite{einstein1935can,horodecki2009quantum}, and it is considered a resource in the context of quantum information science\cite{chitambar2019quantum,wilde2013quantum,bouwmeester2000physics,schleich2016quantum}.  It plays a central role in non-trivial quantum protocols and algorithms and its generation is considered to be one of the essential capabilities when constructing a quantum computer \cite{ladd2010quantum,mermin2007quantum,preskill2012quantum}. While entanglement is most often associated with the microscopic world, it has been also shown to be abundantly present in quantum many-body systems \cite{amico2008entanglement,its2005entanglement,zhang2005thermal,latorre2009short,johnthomas2019}. Atomic gases are a particularly fascinating physical platform for observing many-body entanglement, due to the high level of controllability and low decoherence\cite{hammerer2010quantum,lukin2000entanglement}.  One of the most elementary type of entangled states for an atomic gas are spin squeezed states, where particular observables are reduced below the standard quantum limit \cite{sorensen2001many,hald1999,kuzmich2000generation,esteve2008squeezing,kunkel2018spatially,fadel2018spatial}, and has numerous applications in quantum metrology \cite{gross2012spin,Giovannetti,giovannetti2004quantum,toth2014quantum,giovannetti2011advances,bao2020spin,Sekatski2017quantummetrology}. It has also been observed that Bell violations \cite{bell1964einstein,freedman1972experimental,aspect1982experimental}, which are a stronger form of quantum correlations in the quantum quantifier hierarchy \cite{adesso2016measures,ma2019operational}, can be generated in Bose-Einstein condensates \cite{schmied2016bell}. More exotic types of quantum many-body state can be generated through techniques to perform quantum simulation with a variety of applications \cite{lloyd1996universal,buluta2009quantum,timquantumoptics2020,jane2003simulation,you2017multiparameter,horikiri2016high,lewenstein2007ultracold,lumingduan2021}.

While most of the work relating to entanglement in atomic ensembles has been focused on entanglement that exists between atoms in a single ensemble \cite{gross2012spin,hammerer2010quantum}, works extending this to two or more spatially separate ensembles have also been investigated both theoretically and experimentally. The first experimental demonstration of entanglement between atomic gases was observed in paraffin-coated hot gas cells \cite{julsgaard2001experimental}.  In the scheme, a quantum nondemolition (QND) measurement was performed by beams sequentially illuminating the two gas cells.  This entanglement was used to demonstrate teleportation between two atomic clouds \cite{krauter2013deterministic}, for  continuous variable quantum observables \cite{braunstein2005quantum}.  For Bose-Einstein condensates, currently no experimental demonstration of entanglement between two separate atomic clouds has been performed.  The closest demonstration has been the observation of entanglement between spatially separate regions of a single cloud \cite{fadel2018spatial,kunkel2018spatially,lange2018entanglement,li2013entanglement}.  Numerical and theoretical schemes for entanglement between BECs have been proposed, using a variety of techniques ranging from cavity QED \cite{pyrkov2013,rosseau2014entanglement,ortiz2018adiabatic,hussain2014geometric,abdelrahman2014coherent}, Rydberg excitations \cite{idlas2016}, state dependent forces \cite{treutlein2006}, adiabatic transitions \cite{ortiz2018adiabatic}, and others \cite{abdelrahman2014coherent,oudot2017optimal,Jing_2019}. Such entanglement is fundamental to performing various quantum information tasks based on atomic ensembles, such as quantum teleportation \cite{pyrkov2014quantum,pyrkov2014full}, remote state preparation and clock synchronization \cite{ilo2018remote,manish2021}, and quantum computing \cite{byrnes2012macroscopic,byrnes2015macroscopic}.  

In this paper, we present a measurement operator formalism for QND measurement induced entanglement between two atomic ensembles.  In a previous paper, we developed an exact theory to describe the effect of the QND induced entanglement \cite{aristizabal2021quantum} (see also Refs. \cite{pettersson2017light,ilo2014theory}). The theory is exact in the sense that no approximation is made in terms of the total spin of the atomic ensemble.  In many approaches to QND measurements, only low-order spin correlators are used to capture the dynamics of the measurement, such as working within a Holstein-Primakoff approximation \cite{julsgaard2001experimental,kuzmich2000generation,duan2000quantum,serafin2021nuclear,tsang2012evading}. In our approach, the full wavefunction of the atomic spin can be calculated, due to the exactly solvable dynamics of the QND interaction.  Here, we show how the theory of Ref. \cite{aristizabal2021quantum} can be written in terms of measurement operators, and consider particularly the limiting case where it can be used to construct a projection operator. In Ref. \cite{aristizabal2021quantum} it was noted that just a sequence of two QND measurements can improve the spin correlations.  We develop a general theory of such a sequence of QND measurements (``stroboscopic measurements'') and  analyze the types of states that are generated. Such stroboscopic measurements have been used in the single atomic ensemble case to drive the state towards a macroscopic singlet state \cite{behbood2013real,behbood2014,behbood2013feedback}.  We show that due to the special symmetries that are present in the stroboscopic sequence, it is possible to find the exact states that the system converges in the limit of many stroboscopic projections. Such states are entangled states and thus the scheme can be used as the way of entanglement preparation between atomic ensembles.  

This paper is structured as follows.  In Sec. \ref{sec2} we review the theory of Ref. \cite{aristizabal2021quantum} and introduce the basic system that we are dealing with.  In Sec. \ref{sec3} we introduce a theory of POVMs for the QND measurement, and show that in a particular limiting case this can be used to construct a projection operator.  In Sec. \ref{sec:purestate} we analyze the case of multiple sequential (or stroboscopic) QND measurements.  Here we derive some key mathematical relations which simplify the analysis.  In Sec. \ref{sec:mixedstate} we formulate the projection sequence in a probabilistic framework in terms of density matrices.  In Sec. \ref{vbasistates} we show the properties of the states that the projection sequence converges to.  Finally, in Sec. \ref{sec5} we summarize our results. Some parts of this paper go into the mathematical detail of the measurement operator sequence.  For the reader disinterested in such details, the discussion of Sec. \ref{sec:purestate}C, \ref{sec:purestate}D may be skipped and the results of Lemma 1, 2, and Theorem 1 may be used as mathematical results.

\section{QND induced entanglement}
\label{sec2}

In this section we first briefly review the theory developed in Ref. \cite{aristizabal2021quantum} for producing QND induced entanglement between two atomic ensembles.

\begin{figure}[t]%
\includegraphics[width=\linewidth]{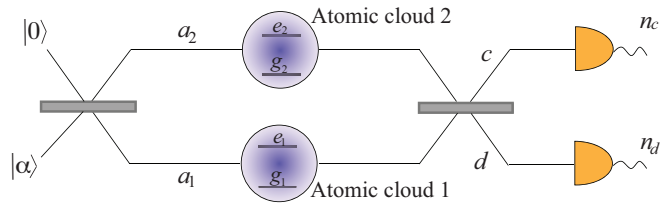}
\caption{QND entangling scheme: An initial coherent light pulse $|\alpha\rangle$ is used to entangle two atomic clouds via the ac stark shift arranged in Mach-Zehnder interferometer configuration. Each of the beamsplitters are 50/50.  The detection of photon outcomes $n_c, n_d$ collapses the state of atomic clouds into entangled states.
}
\label{fig1}%
\end{figure}

\subsection{Definitions}

First consider the states of the atomic gas clouds (see Fig. \ref{fig1}).  We consider each atom within the ensembles to be occupied by one of two internal states.  For example, the states may be two hyperfine ground states of an atom (e.g. $ F=1, m=-1$ and $ F=2, m = +1$ for $^{87}$Rb). The motional degrees of freedom of the atom are decoupled from the spin and may be neglected.  In the case of a BEC, we may define bosonic annihilation operators $e_j,g_j$ for the two internal states respectively, and $ j \in \{1,2 \}$ label the two atomic ensembles \cite{timquantumoptics2020}.  For a collection of $ N $ atoms in each BEC, the initial state of the atomic cloud can be described by 
\begin{align}
		|\psi \rangle = \sum_{k_1,k_2 = 0}^{N} \psi_{k_1 k_2} |k_1, k_2\rangle ,
		\label{generalstate}
\end{align}
where we have defined the Fock states on the $ j$th atomic ensemble as
\begin{align}
|k \rangle =\frac{(e_j^{\dagger})^k(g_j^{\dagger})^{N-k}}{\sqrt{k!(N-k)!}}|\text{vac}\rangle , 
\label{fockstates}
\end{align}
and 
\begin{align}
|k_1, k_2\rangle = | k_1 \rangle \otimes | k_1 \rangle .  
\label{twofocks}
\end{align}
Here the Fock states obey $ \langle k | k' \rangle = \delta_{k k'} $ and the coefficients $ \sum_{k_1, k_2} | \psi_{k_1 k_2} |^2  = 1 $ are normalized.  

For uncondensed thermal atomic gases, in general there are $ 2^N $ possible spin configurations per ensemble, instead of the $ N +1 $ states as defined in (\ref{fockstates}).  However, if the initial state and all applied Hamiltonians are completely symmetric under particle interchange on a single ensemble, there is a mathematical equivalence between the BEC description and the thermal ensemble  \cite{timquantumoptics2020}. Since we will work in the completely symmetric subspace, our results will be equally valid for the thermal atomic case, despite using the bosonic notation.  

The collective spin operators on the $ j $th ensemble  are defined by
\begin{align}
 S^x_j & =e_j^\dagger g_j+e_j^\dagger g_j \nonumber \\
S^y_j & =-ie_j^\dagger g_j +ig_j^\dagger e_j \nonumber \\
S^z_j & =e_j^\dagger e_j-g_j^\dagger g_j,
\end{align}
obeying commutation relations $[S^j,S^k]=2i\epsilon_{jkl}S^l$, where $\epsilon_{jkl}$ is the completely anti-symmetric Levi-Civita tensor.
Spin coherent states, which are completely polarized spin configurations with Bloch sphere angles $ \theta, \phi $ are defined as
\begin{align}
| \theta, \phi \rangle \rangle = \frac{1}{\sqrt{N!}} \left( e^\dagger \cos \frac{\theta}{2}  + g^\dagger e^{i\phi} \sin \frac{\theta}{2}  \right)^N | \text{vac} \rangle  .  
\end{align}
Expanding the spin coherent state we may equally write this in terms of Fock states
\begin{align}
    | \theta, \phi \rangle \rangle = \sum_{k=0}^N \sqrt{N \choose k}
  e^{i(N-k) \phi}   \cos^k \frac{\theta}{2} \sin^{N-k} \frac{\theta}{2}  |k \rangle . 
\end{align}

\subsection{QND entangled wavefunction}

The QND entangling scheme is shown in Fig. \ref{fig1}.  
Here, coherent light is arranged in a Mach-Zehnder configuration and the two atomic gases are placed in each arm of the interferometer. Preparing the atoms in the initial state (\ref{generalstate}), the light interacts with the atomic spins via the QND Hamiltonian \cite{ilo2014theory}, 
\begin{align}
H = \frac{\hbar\Omega}{2}(S^z_{1}-S^z_{2}) ( a^\dagger_1 a_1-a^\dagger_2 a_2) ,
\label{qndhamiltonian}
\end{align}
where $ a_1, a_2 $ denote the bosonic annihilation operators of the light in the two arms of the interferometer.  After interacting with the atoms, the two modes are interfered via the second beam splitter and the photons are detected.

The above sequence modulates the quantum state of the atoms due to the atom-light entanglement that is generated by the QND interaction.  This can be evaluated exactly due to the diagonal form of (\ref{qndhamiltonian}).  We refer the reader to Ref. \cite{aristizabal2021quantum} for further details and present only the final result.  The final unnormalized state after detection of $ n_c, n_d $ photons in modes $c,d $ respectively is 
\begin{align}
	| \widetilde{\psi}_{n_c,n_d}(\tau)\rangle = \sum_{k_1,k_2 = 0}^{N} \psi_{k_1,k_2} C_{n_c,n_d}[(k_1-k_2)\tau]  |k_1, k_2\rangle ,
	\label{modfun}
\end{align}
where we defined the function 
\begin{align}
C_{n_c,n_d}(\chi) = \frac{\alpha^{n_c+n_d}e^{-|\alpha|^2/2}}{\sqrt{n_c!n_d!}}\cos^{n_c}(\chi)\sin^{n_d}(\chi) . 
		\label{modulatingfunc}
\end{align}
Here, $ \alpha $ is the amplitude of the coherent light entering the first beamsplitter in Fig. \ref{fig1} and $ \tau = \Omega t $.  The probability of obtaining a photonic measurement outcome $n_c, n_d $ is
\begin{align}
p_{n_c,n_d} (\tau) & = \langle  \widetilde{\psi}_{n_c,n_d}(\tau)	| \widetilde{\psi}_{n_c,n_d}(\tau)\rangle  \nonumber \\
& = \sum_{k_1, k_2 =0}^N | \psi_{k_1 k_2}  C_{n_c,n_d}[(k_1-k_2)\tau] |^2 . 
\label{probability}
\end{align}
We note that the $ C $-functions are normalized according to
\begin{align}
\sum_{n_c,n_d = 0}^{\infty}|C_{n_c,n_d}(\chi)|^2 = 1  . 	
\label{totalprob}
\end{align}
The most likely photon number counting outcomes are centered around
\begin{align}
n_c+n_d \approx |\alpha|^2,
\end{align}
since the input coherent state has an average photon number of $ |\alpha|^2 $ and the remaining operations are photon number conserving.

The $C$-functions can be approximated for bright coherent light regime $ |\alpha | \gg 1 $ as
\begin{align}
C_{n_c n_d} ( \chi ) \approx & \frac{\alpha^{n_c+ n_d} e^{- |\alpha |^2/2}}{\sqrt{(n_c+n_d)!}}
\frac{\text{sgn} ( \cos^{n_c} ( \chi) \sin^{n_d} ( \chi) )}{ ( \frac{\pi}{2} (n_c + n_d) \sin^2 2 \chi)^{1/4}} \nonumber \\
\times & \exp \Big[ - \frac{(n_c+n_d)}{\sin^2 2 \chi} \big( \sin^2 \chi - \frac{n_d}{n_c+n_d} \big)^2 \Big] ,  
\label{ndnonzeroapprox}
\end{align}
for $ n_d > 0 $.  In the case of  $ n_d = 0 $, the $C$-function is better approximated as 

\begin{align}
C_{n_c n_d =0 } (\chi) \approx \frac{\alpha^{n_c} e^{- |\alpha|^2/2}}{ \sqrt{n_c !} } e^{- \frac{n_c}{2} \sin^2 \chi }  . 
\label{ndzeroapprox}
\end{align}

\subsection{Example} 
\label{sec:example1}

To see how entanglement is generated by the QND scheme, let us choose an initial state for the atoms that is polarized in the $ S^x$-direction
\begin{align}
|\psi_0 \rangle & = |\frac{\pi}{2},0\rangle \rangle |\frac{\pi}{2},0\rangle \rangle \nonumber \\
& = \frac{1}{2^N} \sum_{k_1,k_2 = 0}^{N} \sqrt{{N \choose k_1} {N \choose k_2}   } |k_1 , k_2\rangle.
		\label{initialstatexx}
\end{align}
Now consider the $ n_c \sim |\alpha|^2, n_d = 0 $ photonic measurement outcome, which is a high probability result for short interaction times $ \tau \sim  1/N $. Using the approximation (\ref{ndzeroapprox}) in (\ref{modfun}) with $ \sin \chi \approx \chi $,  we obtain 
\begin{align}
| \widetilde{\psi}_{n_c,n_d}(\tau)\rangle \propto \sum_{k_1,k_2 = 0}^{N}  \sqrt{{N \choose k_1} {N \choose k_2} } 
e^{- \frac{n_c \tau^2 }{2} (k_1-k_2)^2 } |k_1 , k_2\rangle .
	\label{modfunapprox}
\end{align}
In the regime $ n_c \tau^2 \approx |\alpha \tau|^2> 1 $, the Gaussian factor suppresses terms except for $ k_1 = k_2 $.  This takes the form of an entangled state \cite{aristizabal2021quantum,kitzinger2020}.

\section{Quantum measurement theory for QND induced entanglement}
\label{sec3}

In quantum mechanics, a measurement is represented by a positive operator-valued measure (POVM) \cite{Nielsen:2011}, of which projective measurements are a special case. The QND entangling procedure may be viewed as a particular way of measuring the atomic states such that it collapses the state onto an entangled state. In this section, we introduce a POVM based theory of QND measurements, and its associated relations to connect the photonic readouts to particular projection operators.

\subsection{QND POVM operators}

According to the QND entangling protocol described in the previous section,  the initial wave function (\ref{generalstate}) is modulated by an extra factor of $C_{n_c,n_d}[(k_1-k_2)\tau]$ and the final state becomes (\ref{modfun}). 
An efficient way of summarizing this procedure is to define the measurement operator
\begin{align}
M_{n_c n_d}(\tau) =\sum_{k_1,k_2 = 0}	C_{n_c,n_d}[(k_1-k_2)\tau]|k_1,  k_2\rangle \langle k_1 , k_2| .
\label{povmdef}
\end{align}
According to the theory of quantum measurements, the resulting state after the measurement is
\begin{align}
|\widetilde{\psi}_{n_c n_d} (\tau) \rangle = M_{n_c n_d}(\tau) | \psi \rangle
\end{align}
and the probability of this outcome is 
\begin{align}
p_{n_c n_d}(\tau) = \langle \psi |  M_{n_c n_d}^\dagger (\tau)  M_{n_c n_d} (\tau)  | \psi \rangle
\end{align}
in agreement with (\ref{modfun}) and (\ref{probability}) respectively. Since  $  M_{n_c n_d}^\dagger (\tau)  M_{n_c n_d} (\tau) $ is a positive operator and we can evaluate
\begin{align}
\sum_{n_c,n_d }M^{\dagger}_{n_c n_d} (\tau) M_{n_c,n_d} (\tau) = I
\label{completeness}
\end{align}
due to the relation (\ref{totalprob}), we may say that $  M_{n_c n_d} (\tau) $ satisfies the definition of being a POVM.

\begin{figure}[t]%
		\includegraphics[width=\linewidth]{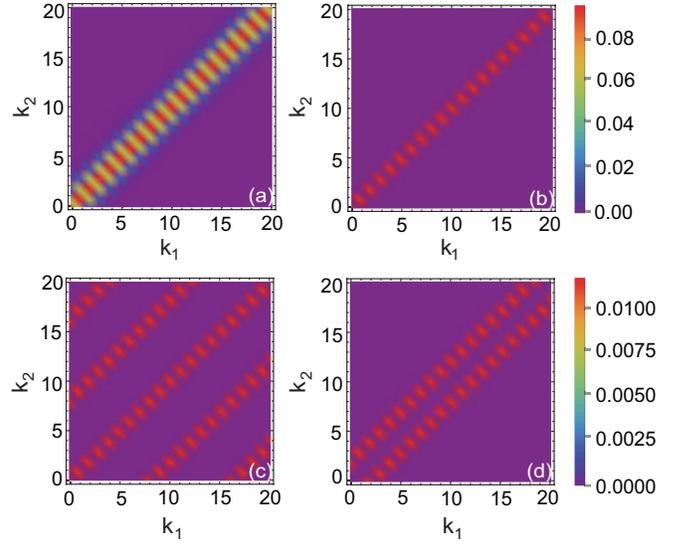}
		\caption{Plot of $C_{n_c,n_d}[(k_1-k_2)\tau]$ in (\ref{modulatingfunc}) with different Fock state pairs for different photon detection outcomes and interaction times. Interaction times are (a) $ t = \frac{\pi}{8N}$, (b) $  t =  \frac{\pi}{2N}$, (c) $  t =\frac{\pi}{8}$, (d) $ t = \frac{\pi}{2N} $.  The measurement outcomes are (a)(b)(c) $ n_c=|\alpha|^2, n_d = 0$, (d)  $ n_c = |\alpha|^2-\alpha, n_d = \alpha $.  For all calculations, number of atoms per BEC is $ N = 30 $ and  $\alpha = 50 $. The scale  for (a)-(c) are the same and shown next to (b).  The  scale for (d) is shown separately and is shown adjacent to the figure.  }
		\label{fig2}%
\end{figure}

The type of the measurement induced by $  M_{n_c n_d}(\tau)  $ depends greatly upon the interaction time $ \tau$ and the outcomes $ n_c, n_d $.  First let us look at the effect of the interaction time $ \tau $.  In Fig.  \ref{fig2}, we plot the modulating function $C_{n_c,n_d}[(k_1-k_2)\tau]$ for various photon detection outcomes $n_c, n_d$ and interaction times $ \tau $. First setting $ n_d = 0 $ and comparing $ \tau = \pi/8N, \pi/2N $ (Figs. \ref{fig2}(a)(b)) we see that for shorter times than $ \tau = \pi/2N $ the function along the diagonal $ k_1 = k_2 $ is broadened.  For interaction times longer than $ \tau = \pi/2N $, additional diagonal lines occur (Fig. \ref{fig2}(c)) according to the location of the peak of the Gaussian in (\ref{ndnonzeroapprox})
\begin{align}
    \sin^2 [ (k_1 - k_2) \tau ]= \frac{n_d}{n_c+n_d} . 
    \label{sinerelation}
\end{align}
For measurement outcomes detecting $n_d>0$ as in Fig. \ref{fig2}(d), we see the correlations are offset following the relation (\ref{sinerelation}).  

To explicitly see the effect of the various $ n_c, n_d $ outcomes, we plot the probability of the measurement outcome
\begin{align}
p_{n_c n_d}( \tau)  &  = \langle k, k+\Delta |  M_{n_c n_d}^\dagger ( \tau)   M_{n_c n_d} ( \tau)   | k , k+\Delta \rangle \nonumber \\
& = |C_{n_c,n_d} ( \Delta \tau) |^2 . 
\end{align}
This gives the probability of various $ n_c, n_d$ outcomes for a state that differs in Fock state by $ \Delta =k_2 - k_1 $. In Fig. \ref{fig3} we see a probability curve that is centered around (\ref{sinerelation}).  For a time $ t = \pi/2N$, there is a one-to-one relation between the measurement readout $ n_c, n_d$ and the magnitude of the Fock number difference $ |\Delta| $  (Fig.  \ref{fig3}(a)).  At this time, the two Fock state differences that have a high probability are
\begin{align}
   \Delta & = \pm \frac{1}{\tau} \sin^{-1} \sqrt{\frac{n_c}{n_c + n_d}} . 
   \label{deltancnd}
\end{align}
For longer times, the relationship is no longer one-to-one (Fig.  \ref{fig3}(b)), and corresponds to the multiple diagonal lines seen in Fig. \ref{fig2}(c).  In order to have a sharply defined projections without additional peaks in $\Delta $, we henceforth consider the time $ \tau = \pi/2N $.

\begin{figure}[t]%
\includegraphics[width=\linewidth]{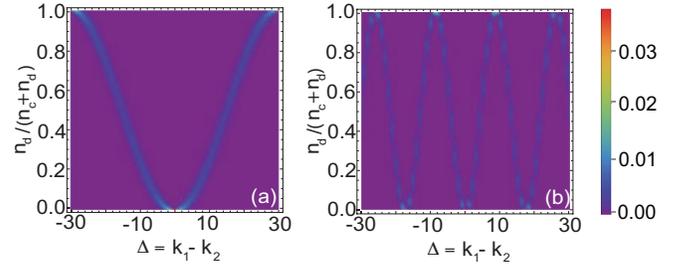}
\caption{Relation of various photon detection outcomes $n_c, n_d$ with offset in Fock states $\Delta = k_1-k_2$ for (a) $ \tau = \frac{\pi}{2N} $, (b) $ \tau = \frac{1}{\sqrt{N}} $ for $N = 30,\alpha=10$. }
\label{fig3}%
	\end{figure}

\subsection{QND projection operators}
\label{sec:qndprojectors}

In the limit that the intensity of coherent light $ |\alpha |^2 $ is very bright such that $  |\alpha \tau |^2 > 1 $, the Gaussian function in  (\ref{ndnonzeroapprox}) is sharply defined and strongly suppresses values of $ k_1 - k_2 $ away from (\ref{sinerelation}). 
In this limit, when two such measurements are made in succession, the net result is a projection operator. Taking the interaction time $ \tau  = \pi/2N $, there are two values of $ k_1 - k_2 $ where the projections occur (Fig. \ref{fig2}(d)), as given by (\ref{deltancnd}).  We may then approximate the POVM according to
\begin{align}
   M_{n_c n_d} ( \tau= \frac{\pi}{2N}  ) \approx  \mathcal{M} _\Delta \hspace{1cm} ( |\alpha \tau |^2  \gg 1 ) 
\end{align}
%
%\begin{align}
 %  M_{n_c n_d} ( \tau= \frac{\pi}{2N}  ) \approx  P_\Delta \hspace{1cm} ( |\alpha \tau |^2  \gg 1 ) 
%\end{align}
%
where the $ \Delta $ and $ n_c, n_d $ are related according to (\ref{deltancnd}), and we define
\begin{align}
\mathcal{M}_\Delta =&\frac{1}{2^{ \delta_{\Delta} }} \Big( \sum_{k=0}^{N-\Delta}|k,k+\Delta\rangle \langle k,k+\Delta| \nonumber\\
& \pm \sum_{k'=\Delta}^N |k',k'-\Delta\rangle \langle k',k' -\Delta|  \Big) . \nonumber 
\end{align}
A double sequence of such POVMs then gives the projection operator
\begin{align}
    P_\Delta =& \mathcal{M}_\Delta \mathcal{M}_\Delta\nonumber\\
 =&\frac{1}{2^{ \delta_{\Delta} }} \Big( \sum_{k=0}^{N-\Delta}|k,k+\Delta\rangle \langle k,k+\Delta| \nonumber\\
& +\sum_{k'=\Delta}^N |k',k'-\Delta\rangle \langle k',k' -\Delta|  \Big) . 
\label{projzbasis}
\end{align}
Here $ \delta_{\Delta} $ is the Kronecker delta which is 1 if $ \Delta = 0 $ and 0 otherwise.  

The above projection operators are defined with respect to Fock states that are in the $ S^z $ basis.  We may equally define the Fock states in a different basis \cite{timquantumoptics2020}
\begin{align}
|k\rangle^{(\theta,\phi)} =  {\cal U}  (\theta, \phi) |k\rangle^{(z)}
\end{align}
where 
\begin{align}
  {\cal U} (\theta, \phi) =  e^{-i (S^z_1+S^z_2) \phi /2}  e^{-i (S^y_1 + S^y_2) \theta/2} 
  \label{unitaryrotation}
\end{align}
and $ |k\rangle^{(z)} $ are the same Fock states as defined in (\ref{fockstates}).  The projection operators may then be defined with respect to Fock states in a different basis
\begin{align}
P_\Delta^{(\theta,\phi)} =  {\cal U}  (\theta, \phi) P_\Delta^{(z)}  {\cal U}^\dagger (\theta, \phi)
\label{rotatedproj}
\end{align}
where $ P_\Delta^{(z)} $ is the same projector as in (\ref{projzbasis}), but we explicitly specified the basis with the $ ^{(z)} $ label.  For the case that $ \theta = \pi/2 $ and $ \phi = 0 $, the unitary rotation transforms the $ S^z$ eigenstates to $ S^x $ eigenstates and we define
\begin{align}
P_\Delta^{(x)} =   {\cal U} (\frac{\pi}{2}, 0 ) P_\Delta^{(z)}  {\cal U}^\dagger (\frac{\pi}{2}, 0) .
\label{pxprojector}
\end{align}

\subsection{Properties of projection operators}

Here we list the properties of projection operators  (\ref{rotatedproj}).  These also apply to the specific cases (\ref{projzbasis}) and (\ref{pxprojector}).   
\begin{enumerate}
\item The projection operators are idempotent and orthogonal
\begin{align}
	P_{\Delta}^{(\theta,\phi)} P_{\Delta'}^{(\theta,\phi)}  = \delta_{\Delta \Delta'} P_{\Delta}^{(\theta,\phi)} . 
	\label{idemortho}
\end{align}
	\item Projection operators are Hermitian 
	\begin{align}
		(P^{(\theta,\phi)}_\Delta)^\dagger=P^{(\theta,\phi)}_\Delta . 
	\end{align}
\item Projection operators are complete 
	\begin{align}
		\sum_{\Delta}(P^{(\theta,\phi)}_\Delta)^\dagger P^{(\theta,\phi)}_\Delta= 	\sum_{\Delta} P^{(\theta,\phi)}_\Delta = I.  
	\end{align}
	\item The eigenvalues and eigenvectors of projection operators are
	\begin{align}
P_{\Delta}^{(\theta,\phi)} |k, k+\Delta' \rangle^{(\theta,\phi)} = \delta_{\Delta \Delta'} |k, k+\Delta' \rangle^{(\theta,\phi)}  .
\label{eigenvaluesproj}
	\end{align}
	\end{enumerate}
The proofs of these properties are straightforward and can be verified by substituting the definitions.

\subsection{Example}

We show that the projection operator can generate entanglement between the two atomic ensembles by applying it to the initial state (\ref{initialstatexx}).  For example, for the outcome $ \Delta = 0 $ we have
\begin{align}
P_{\Delta = 0 }^{(z)} | \psi_0 \rangle = \frac{1}{2^N} 
\sum_{k=0}^N {N \choose k } | k, k \rangle .
\label{projexampleonez}
\end{align}
Similarly to that obtained in Sec. \ref{sec:example1}, we obtain an entangled state between the two BECs.  

We point out that the entanglement generation depends upon the preparation of a suitable initial state.  From (\ref{eigenvaluesproj}) it is evident that the projection operator applied to a product state of two Fock states leaves it unchanged.  Since a product state does not possess entanglement between the BECs, the operation does not produce entanglement in this case.  In this sense, a single QND measurement should not be considered an equivalent of Bell measurement, for which any initial state and measurement outcome results in an entangled state.  It is better described as a projective operation which can result in an entangled state for suitable initial states.  In the next sections we show how a sequence of projections can drive an arbitrary initial state into an entangled state.

\section{Sequential QND projections: pure state analysis}
\label{sec:purestate}

In Ref. \cite{aristizabal2021quantum} a two-pulse scheme for improved entanglement generation was analyzed.  In the scheme, first the initial state (\ref{initialstatexx}) is prepared, and a QND measurement is performed, to yield a state similar to  (\ref{projexampleonez}).  Then a unitary rotation $ e^{i S^x_1 \pi/4} e^{-i S^x_2 \pi/4} $ was performed and the QND measurement was repeated.  This was found to produce a more strongly entangled state than using only a single QND measurement.  Such a sequence of pulses have been used experimentally in several studies to enhance the entanglement in single atomic ensembles \cite{behbood2013real,behbood2013feedback,behbood2014,vasilakis2015generation}.  We show in this section that the technique is equally applicable in the two-ensemble case, and develop a theory to evaluate the result of multiple projections.

\subsection{Multiple QND measurements}

First let us define the multiple QND measurement scheme in terms of the formalism we have introduced so far.  We work in the regime $ |\alpha \tau |^2 > 1 $ such that the QND measurements can be described using the projection operators introduced in Sec. \ref{sec:qndprojectors}.  Let us define a particular QND projection sequence as 
\begin{align}
T_{\vec{\Delta}_{2L} } & =  P_{\Delta^x_L}^{(x)}P_{\Delta^z_L}^{(z)} \dots  P_{\Delta^x_2}^{(x)}P_{\Delta^z_2}^{(z)} P_{\Delta^x_1}^{(x)}P_{\Delta^z_1}^{(z)} \nonumber \\
& = \prod_{l=1}^L   \Big[ P_{\Delta^x_l}^{(x)}P_{\Delta^z_l}^{(z)}  \Big] .  
   \label{multprojeven}
\end{align}
where the projections in the $z$- and $x$-basis are defined in (\ref{projzbasis}) and (\ref{pxprojector}) respectively. In the product operator, we take the convention that the order of the projectors is arranged from right to left, for labels running from the lower index to the upper index.  Each projection is made in an alternating basis switching between $ z$ and $x$. 
The sequence consists of $ L $ repetitions of projections in the $ z$ and $x$ basis. We take the convention that a sequence always starts with a projection in the $z$-basis. Equation (\ref{multprojeven}) is an example where there are an even number of projections (a total of $2L$ projections) in total. Thus in this case the last projector will be in the $x$-basis. The outcomes are specified by an ordered list
\begin{align}
\vec{\Delta}_{2L} = (\Delta_1^z,\Delta_1^x,\Delta_2^z,\Delta_2^x,\dots,\Delta_L^z,\Delta_L^x) .
\label{deltavectoreven}
\end{align}
where the subscript $2L$ shows the number of projections that are made. Physically, these can be interpreted as the random outcomes associated with a particular projection sequence.  

For an odd number of projections, we have instead
\begin{align}
T_{\vec{\Delta}_{2L+1} } & = P_{\Delta^z_{L+1}}^{(z)} P_{\Delta^x_L}^{(x)}P_{\Delta^z_L}^{(z)} \dots  P_{\Delta^x_2}^{(x)}P_{\Delta^z_2}^{(z)} P_{\Delta^x_1}^{(x)}P_{\Delta^z_1}^{(z)} \nonumber \\
& = P_{\Delta^z_{L+1}}^{(z)}\prod_{l=1}^L  \Big[ P_{\Delta^x_l}^{(x)}P_{\Delta^z_l}^{(z)}  \Big] .  
   \label{multproj}
\end{align}
Here a particular sequence is specified by $ 2L +1 $ parameters
\begin{align}
\vec{\Delta}_{2L+1} = (\Delta_1^z,\Delta_1^x,\Delta_2^z,\Delta_2^x,\dots,\Delta_L^z,\Delta_L^x,\Delta_{L+1}^z ) .
\label{deltavector}
\end{align}
We shall see that the parity of the number of projections (i.e. whether it is even or odd) will make a large difference to the final state.  When discussing a property without any parity dependence, we will drop the subscript on  $ \vec{\Delta}$ for brevity. %The odd number of projection is in fact the more common case that we will examine, hence we will consider (\ref{multproj}) and (\ref{deltavector}) to be the definition of $ T_{\vec{\Delta} } $ unless specified. 

Starting from an initial state $ |\psi_0\rangle $, after a particular projection sequence we obtain an unnormalized state (denoted by the tilde)
\begin{align}
  | \tilde{\psi}_{\vec{\Delta} } \rangle & = T_{\vec{\Delta} }|\psi_0\rangle .  
  \label{multitpsi}
\end{align}
The probability of this particular outcome labeled by $ \vec{\Delta} $ occurring is
\begin{align}
p_{\vec{\Delta} } = \langle \tilde{\psi}_{\vec{\Delta} }| \tilde{\psi}_{\vec{\Delta} }\rangle ,  
\label{probdelta}
\end{align}
where the probabilities satisfy
\begin{align}
     \sum_{\vec{\Delta}} p_{\vec{\Delta} } \equiv \sum_{\Delta_1^z} \sum_{\Delta_1^x} \dots \sum_{\Delta_{L+1}^z} p_{\vec{\Delta} } =1 .  
\label{totalprobone}
\end{align}
We may visualize the sequence of projections as shown in Fig. \ref{fig4}. Each projection occurs randomly and yields in general a different state.  Successive projections yield new states that depend upon the past projection outcomes.

\begin{figure}[t]%
\includegraphics[width=\linewidth]{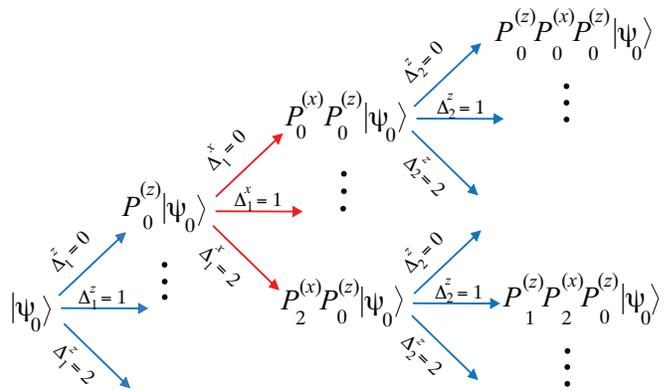}
\caption{Tree diagram corresponding to the multiple QND projection sequence $ T_{\vec{\Delta}} $.  Each branch of the tree corresponds to different projection outcomes labelled by (\ref{deltavector}).  
The probability of a particular outcome is given by (\ref{probdelta}).
}
\label{fig4}%
\end{figure}

We note that the sequence of projection operators $ T_{\vec{\Delta} } $ is not itself a projection operator.  Specifically, the property of idempotence and orthogonality (\ref{idemortho}) does not hold.  They are however complete and satisfy 
\begin{align}
   \sum_{\vec{\Delta}} T_{\vec{\Delta}}^\dagger  T_{\vec{\Delta}}  =    \sum_{\vec{\Delta}} T_{\vec{\Delta}} = I,
\end{align}
which may be used to show (\ref{totalprobone}).

\subsection{Example: convergence to maximally entangled state}
\label{sec:maxentangle}

To see the effect of the multiple QND projections, it is illustrative to see a simple example.  Let us consider the particular outcome sequence with an odd number of projections
\begin{align}
  \vec{\Delta}_{2L+1} = (0,0,\dots,0) .    
  \label{allzerodelta}
\end{align}
The final state in this case is 
\begin{align}
| \tilde{\psi}_{\vec{\Delta}_{2L+1}} \rangle = P_{0}^{(z)} \Big( P_{0}^{(x)}P_{0}^{(z)}  \Big)^L | \psi_0 \rangle . 
\label{allzeropsi}
\end{align}
The effect of the projector $ P_{0}^{(z)} $ is to remove all elements from the wavefunction except $ k_1 =k_2  $ in (\ref{generalstate}).  Hence after each projector  $ P_{0}^{(z)} $  the state is guaranteed to be in a state of the form
\begin{align}
\sum_{k=0}^N \psi_{kk} | k, k \rangle  .
\label{projectedform}
\end{align}

\begin{figure}[t]%
\includegraphics[width=\linewidth]{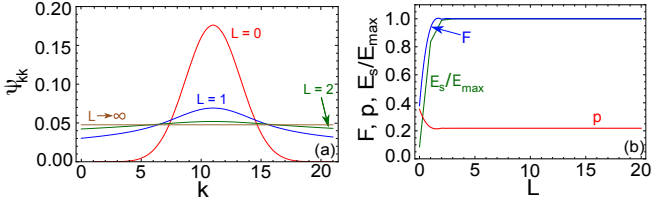}
\caption{Analysis of multiple QND projections: (a)The resulting wavefunction after multiple QND projections of the form (\ref{multitpsi}). The $z$-basis projected wavefunction of the form (\ref{projectedform}) is shown; (b) fidelity, entanglement and probability of the projected wavefunction for $ N = 20 $. For both the figures we have considered $\vec{\Delta}_{2L+1} = (0,0,\dots,0) $.
}
\label{fig5}%
\end{figure}

The effect of multiple projections with $ L $ can then be seen in Fig. \ref{fig5}(a), where we start in the initial state (\ref{initialstatexx}). We see that for $L = 0$ the distribution takes a binomial form, as already seen in (\ref{projexampleonez}).   The largest amplitude occurs at $ k = N/2$. As $ L$ is increased, the state gradually flattens out. This can also be seen in terms of the spin correlations.  In Fig. \ref{fig6} we show the probability of the projected state after multiple QND measurements in various spin bases i.e. $|k\rangle^{(x)}$, $|k\rangle^{(y)}$ and $|k\rangle^{(z)}$ which are eigenstates of the spin operators $S^x$, $S^y$ and $S^z$ respectively. 
The probability of the measurement outcomes $k_1,k_2$ for various Fock states is defined as 
\begin{align}
    p_{l_1l_2} = \frac{|\langle \tilde{\psi}_{\vec{\Delta}_{2L+1}} |\Big(|k\rangle^{(l_1)}\otimes |k\rangle^{(l_2)}\Big)|^2 }{  \langle  \tilde{\psi}_{\vec{\Delta}_{2L+1}} |\tilde{\psi}_{\vec{\Delta}_{2L+1}}  \rangle }, 
    \label{probvariousbases}
\end{align}
where the final projected state is used (\ref{multitpsi}). 
We see that the spin correlations are improved and become more sharply defined, with correlations in the $S^x, S^z$ variables and anticorrelations in $S^y $. 

After a larger number of QND measurement rounds, these spin correlations converge to that of a spin-EPR state (\ref{spineprz}) (see Fig. \ref{fig6}(c)). The improvement in the spin correlations is one of the reasons such a measurement sequence would be performed in practice. The state converges to the maximally entangled state 
\begin{align}
  |\text{EPR}_+ \rangle =   \frac{1}{\sqrt{N+1}} \sum_{k=0}^N | k, k \rangle^{(z)}   .  
  \label{spineprz}
\end{align}
This state has basis invariant properties analogous to Bell states,
and can be written equivalently as \cite{kitzinger2020}
\begin{align}
  |\text{EPR}_+ \rangle =   \frac{1}{\sqrt{N+1}} \sum_{k=0}^N | k, k \rangle^{(x)} .  
\end{align}
This state is an eigenstate of both the $ P_{0}^{(z)} $ and $ P_{0}^{(x)} $ projectors with eigenvalue 1, hence for large $ L$ the state converges to the above state. In Fig. \ref{fig5}(b) we show the fidelity of the state 
\begin{align}
    F =\frac{ | \langle \text{EPR}_+ |\tilde{\psi}_{\vec{\Delta}_{2L+1}}  \rangle|^2}{  \langle  \tilde{\psi}_{\vec{\Delta}_{2L+1}} |\tilde{\psi}_{\vec{\Delta}_{2L+1}}  \rangle }
\end{align}
as a function of $ L$.  We see that the fidelity approaches 1 rapidly.  The probability of obtaining the outcome (\ref{probdelta}) approaches a value of $ \approx 0.22 $ that is for a single projection outcome.  We note that for the outcome (\ref{allzerodelta}), the state converges to the state (\ref{spineprz}) regardless of the initial state.  The convergence properties will however be different for different initial states.

\begin{figure}[t]%
\includegraphics[width=\linewidth]{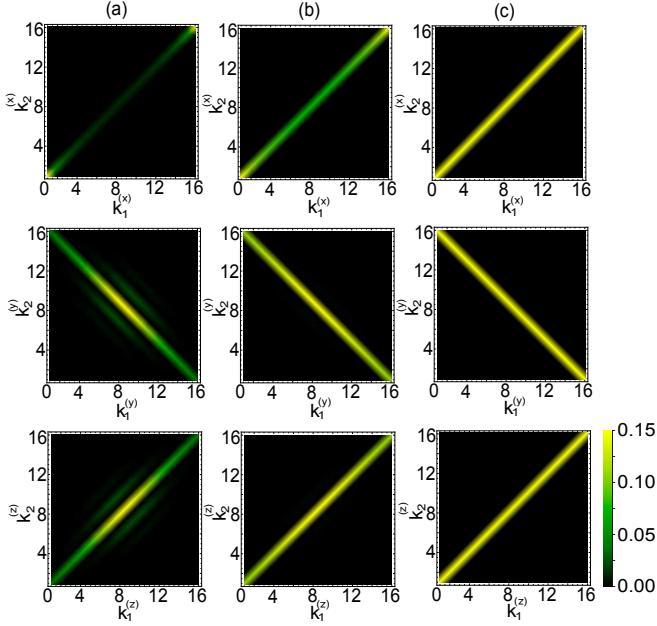}
\caption{Spin correlations after multiple QND projections of the form (\ref{multitpsi}), measured in various bases according to (\ref{probvariousbases}). The number of QND measurement rounds is (a) $L=1$; (b) $L=2$; (c) $L=6$. The number of atoms in each ensemble is  $ N = 15 $ and the measurement outcome is taken as $\vec{\Delta}_{2L+1} = (0,0,\dots,0) $. 
}
\label{fig6}%
\end{figure}

The state (\ref{spineprz}) is a maximally entangled state, and accordingly we may confirm the convergence of entanglement (Fig. \ref{fig5}(b)).  The entanglement can be quantified by the von Neumann entropy, defined as 
\begin{align}
E(\rho) = - \text{Tr} ( \rho \log_2 \rho). 
\label{entanglevonneumann}
\end{align}
where the reduced density matrix on BEC 1 is
\begin{align}
    \rho=Tr_2(| \tilde{\psi}_{\vec{\Delta}_{2L+1}}\rangle \langle \tilde{\psi}_{\vec{\Delta}_{2L+1}}|)
\end{align}
In Fig. \ref{fig5}(b) we have normalized it with maximum entanglement between two BECs $ E_\text{max} = \log_2 (N+1) $. The entanglement increases and saturates at maximum value once the state has a high fidelity with the state (\ref{spineprz}).

\subsection{Singular value decomposition of two projectors}

The example in the previous section illustrates the general behavior of multiple QND projections. The basic behavior is that after a number of QND measurements, the state converges to a fixed state.  This was a relatively simple example, where some known properties of the states could be used to deduce the final state. We now examine the more general case for an arbitrary $ \vec{\Delta} $.  

The following key mathematical result greatly simplifies the subsequent analysis:  
\begin{lemma}
Consider a sequence of two projections
\begin{align}
T_{\Delta^z \Delta^x} = P_{\Delta^x}^{(x)} P_{\Delta^z}^{(z)} 
\label{tmatfortwo}
\end{align}
as defined in (\ref{multprojeven}).  Its singular value decomposition can be written
\begin{align}
T_{\Delta^z\Delta^x}=U \Lambda_{\Delta^z \Delta^x} V^{\dagger}
\label{tsvdform}
\end{align}
where $ U, V$ are unitary matrices that are independent of $ \Delta^z,\Delta^x$, and $ \Lambda_{\Delta^z,\Delta^x} $ is a diagonal matrix up to a permutation of rows and columns, i.e. a matrix with only one non-zero element in each row and column.  
\end{lemma}

\begin{proof}
The full proof is given in Appendix \ref{app:proof}. Here we provide the main steps of the proof. First define the Hermitian matrix
\begin{align}
R_{\Delta^z \Delta^x} & \equiv  T_{\Delta^z \Delta^x}^\dagger T_{\Delta^z\Delta^x} \label{rdefinition}\\
& = P_{\Delta^z}^{(z)}  P_{\Delta^x}^{(x)} P_{\Delta^z}^{(z)} 
\end{align}
From Appendix \ref{app:proof}, %due to the ..... property, 
it follows that for arbitrary choices of the index $ \Delta^z,\Delta^x$ we have
\begin{align}
[R_{\Delta^z \Delta^x} , R_{{\Delta^z}' {\Delta^x}'} ] = 0 .
\label{commuteq}
\end{align}
For any two commuting Hermitian matrices, it is possible to write down a common unitary transformation that diagonalizes both the operators. 
Since all $ R$'s mutually commute, there is a common unitary transformation for all the $ R$'s. Hence
\begin{align}
R_{\Delta^z \Delta^x} = V D_{\Delta^z\Delta^x} V^\dagger ,
\label{rdiagonalizedu}
\end{align}
where $ D $ is a diagonal matrix.  Substituting (\ref{tsvdform}) in (\ref{rdefinition}) we have
\begin{align}
R_{\Delta^z \Delta^x} = V \Lambda_{\Delta^z\Delta^x} \Lambda_{\Delta^z\Delta^x}^\dagger V^\dagger . 
\end{align}
Proving (\ref{commuteq}) for all possible $\Delta^z$ and $\Delta^x$ demonstrates that a common set of eigenvectors diagonalizes the Gram matrix $R_{\Delta^z \Delta^x}$, and the resultant eigenvector matrix therefore forms a consistent set of left singular vectors for $T_{\Delta^z \Delta^x}$.
For the $ U $ unitary matrix, the argument is identical except that one starts with the definition $ R_{\Delta^z \Delta^x} = T_{\Delta^z \Delta^x} T_{\Delta^z\Delta^x}^\dagger$. 
\end{proof}

An example of Lemma 1 is shown in Appendix \ref{app:examplelemma1} for illustration.

\subsection{Evaluation of multiple projectors}

With the aid of Lemma 1, we may write a more general sequence of projections in simpler form. 
\begin{lemma}
The matrix product for a sequence of projections with an odd number of projectors as in (\ref{multproj}) can be written
\begin{align}
    T_{\vec{\Delta}_{2L+1}}& = V\left( \prod_{l=1}^L\Lambda_{\Delta_l^x\Delta_{l+1}^z}^{\dagger}\Lambda_{\Delta_l^x\Delta_l^z}\right) V^{\dagger} ,  
    \label{multprojsimpl}
\end{align}
and for an even number of projectors as in (\ref{multprojeven}),
\begin{align}
    T_{\vec{\Delta}_{2L}}& = U   \Lambda_{\Delta_L^x\Delta_L^z}
    \left( \prod_{l=1}^{L-1}\Lambda_{\Delta_l^x\Delta_{l+1}^z}^{\dagger}
    \Lambda_{\Delta_l^x\Delta_l^z}\right) V^{\dagger} ,  
    \label{multprojsimpleven}
\end{align}
for  $ L \ge 1$. 
Here the unitary matrices $ U, V $  and the matrices
$ \Lambda_{\Delta_l^x\Delta_l^z} $ are the same as given in  (\ref{tsvdform}).   
\end{lemma}
\begin{proof}
Using the idempotence relation we can write (\ref{multproj}) as 
\begin{align}
T_{\vec{\Delta}_{2L+1} } & = \prod_{l=1}^L  \Big[ P_{\Delta^z_{l+1} }^{(z)}  P_{\Delta^x_l}^{(x)} P_{\Delta^x_l}^{(x)}  P_{\Delta^z_l}^{(z)}  \Big] \nonumber \\
& = \prod_{l=1}^L T_{\Delta_{l+1} ^z \Delta_{l}^x}^{\dagger} T_{\Delta_l^z\Delta_l^x} .
   \label{multprojlemma2}
\end{align}
Substituting  (\ref{tsvdform}) into (\ref{multprojlemma2}), the unitary rotations $ U, V $ at all intermediate steps can be removed by $ U^\dagger U = V^\dagger V = I$ since they are common for all $ \Delta^x , \Delta^z $.  This yields (\ref{multprojsimpl}). For (\ref{multprojsimpleven}), use (\ref{multprojsimpl}) with $ L \rightarrow L-1$ and apply  (\ref{tsvdform}) once. 
\end{proof}

From Lemma 2, we see that in a sequence of projections, there is a particular basis where the evolution of the state becomes particularly simple.  The $ \Lambda $ matrices are diagonal up to a permutation of rows and columns, which means that the computation can be performed efficiently.  

To illustrate this, let us now use Lemma 2 to determine the effect of applying the sequential projection for the particular case of an initial state in the $ V$-basis
\begin{align}
| \psi_0 \rangle = | \bm{k} \rangle^{(V)} \equiv   V | \bm{k} \rangle^{(z)} , 
  \label{initialstatevk1k2}
\end{align}
where $ | \bm{k} \rangle^{(z)} $ is defined in (\ref{twofocks}), and we have defined
\begin{align}
\bm{k} = (k_1, k_2) .  
\end{align}
Considering the case with an odd number of projectors first, apply (\ref{multprojsimpl}) to (\ref{initialstatevk1k2}).  This results in the unnormalized state
\begin{align}
   T_{\vec{\Delta}_{2L+1} } | \bm{k} \rangle^{(V)} =   V\left( \prod_{l=1}^L\Lambda_{\Delta_l^x\Delta_{l+1}^z}^{\dagger}
   \Lambda_{\Delta_l^x\Delta_l^z}\right) | \bm{k} \rangle^{(z)} ,
   \label{tappliedk1k2}
\end{align}
where we take the case with an odd number of projections.  Since the matrix $ \Lambda_{\Delta_l^x \Delta_l^z} $ is a permutation matrix up to some coefficients, the effect of each application of the matrix will be to change the labels $ \bm{k} \rightarrow \bm{k}' $, and adjust the normalization of the state.  Let us define the permutation function associated with the matrix $ \Lambda_{\Delta_l^x \Delta_l^z} $ 
\begin{align}
   \bm{k}' & = \pi_{\Delta^x\Delta^z}  (\bm{k} )  . 
\end{align}
For the matrix $ \Lambda_{\Delta^x \Delta^z}^\dagger $, we have the inverse function 
\begin{align}
   \bm{k} & = \pi_{\Delta^x\Delta^z}^{-1}  (\bm{k}' )  . 
\end{align}
Then for the $l$th pair of $ \Lambda $ matrices in (\ref{tappliedk1k2}), the labels shift by
\begin{align}
 \bm{k}_{l}' & = \pi_{\Delta_l^x\Delta_l^z} (  \bm{k}_{l})  \nonumber \\
    \bm{k}_{l+1} & = \pi^{-1}_{\Delta_l^x\Delta_{l+1}^z} ( \bm{k}_{l}').  
    \label{recursivekrelation}
\end{align}
Here $  \bm{k}_{l}' $ are the state labels after application of $  \Lambda_{\Delta_l^x\Delta_l^z} $ and $  \bm{k}_{l+1} $ are the labels after the $ \Lambda_{\Delta_l^x\Delta_{l+1}^z}^{\dagger} $. The initial state label is $ \bm{k}_1 = \bm{k} $ in  (\ref{initialstatevk1k2}). This is a recursive relation by which we can evaluate the final state indices. 

The amplitude of (\ref{tappliedk1k2}) can be evaluated by 
defining the matrix element for this permutation as
\begin{align}
 a_{\Delta^x\Delta^z}(\bm{k}) = \langle \bm{k}' | \Lambda_{\Delta^x\Delta^z} | \bm{k}\rangle .
 \label{afactor}
\end{align}
Application of each  $ \Lambda $ matrix gives a factor as given in (\ref{afactor}).  Thus the final coefficient is 
\begin{align}
    A_{\vec{\Delta}_{2L+1}} ( \bm{k} ) = \prod_{l=1}^{L} \Big[ 
     a_{\Delta^x_l \Delta^z_{l+1} }( \bm{k}_{l+1} )
     a_{\Delta^x_l \Delta^z_l }(\bm{k}_l ) \Big] ,
     \label{atotalcoeff}
\end{align}
and the $ \bm{k}_l $ are evaluated by the recursion relation (\ref{recursivekrelation}).  We also define the result of $ L $ rounds of recursion of the relation (\ref{recursivekrelation}) as %
\begin{align}
    r_{\vec{\Delta}_{2L+1}} (\bm{k}) = \bm{k}_{L+1} .  
    \label{recursionresult}
\end{align}
Then (\ref{tappliedk1k2}) can be evaluated to give the unnormalized state
\begin{align}
 T_{\vec{\Delta}_{2L+1} } | \bm{k}  \rangle^{(V)} = A_{\vec{\Delta}_{2L+1} } ( \bm{k} ) |  r_{\vec{\Delta}_{2L+1} } (\bm{k})  \rangle^{(V)}  .  
 \label{cor1result}
\end{align}

For the case with an even number of projectors, we have
\begin{align}
   T_{\vec{\Delta}_{2L} } | \bm{k} \rangle^{(V)} =   U  \Lambda_{\Delta_L^x\Delta_L^z} \left( \prod_{l=1}^{L-1}  \Lambda_{\Delta_l^x\Delta_{l+1}^z}^{\dagger}
   \Lambda_{\Delta_l^x\Delta_l^z}\right) | \bm{k} \rangle^{(z)} ,
   \label{tappliedk1k2evena}
\end{align}
using (\ref{multprojsimpleven}).  The result of the recursion relation in this case is
\begin{align}
r_{\vec{\Delta}_{2L}} (\bm{k}) = \bm{k}_{L}' .  
\label{recursionresulteven}
\end{align}
and the coefficient is 
\begin{align}
A_{\vec{\Delta}_{2L}} ( \bm{k} ) =a_{\Delta^x_L \Delta^z_L }(\bm{k}_L ) \prod_{l=1}^{L-1} \Big[ 
     a_{\Delta^x_l \Delta^z_{l+1} }( \bm{k}_{l+1} )
     a_{\Delta^x_l \Delta^z_l }(\bm{k}_l ) \Big] .  
     \label{atotalcoeffeven}
\end{align}
The final unnormalized result in this case is
\begin{align}
 T_{\vec{\Delta}_{2L} } | \bm{k}  \rangle^{(V)} = A_{\vec{\Delta}_{2L} } ( \bm{k} ) |  r_{\vec{\Delta}_{2L} } (\bm{k})  \rangle^{(U)}  ,  
 \label{cor1resulteven}
\end{align}
note that the final result is in the $ U $-basis in this case because of the $ U $ in (\ref{tappliedk1k2evena}).   

For (\ref{cor1result}) and (\ref{cor1resulteven}), the normalized state is simply $  |  r_{\vec{\Delta}} (\bm{k})  \rangle^{(V)} $, unless $ A_{\vec{\Delta}} ( \bm{k} ) = 0 $ in which case the projection sequence is an outcome with zero probability.

We are now ready to state our first main result. 
\begin{theorem}
Applying a sequence of projections $ T_{\vec{\Delta} } $ on an arbitrary initial state $ | \phi \rangle $  results in the unnormalized state
\begin{align}
T_{\vec{\Delta}_{2L+1} } | \phi \rangle = \sum_{\bm{k}}  A_{\vec{\Delta}_{2L+1} } ( \bm{k} ) \phi_{\bm{k}}^{(V)} |  r_{\vec{\Delta}_{2L+1} } (\bm{k} )  \rangle^{(V)}
\label{theorem1}
\end{align}
for an odd number of projectors as in (\ref{multproj}) and 
\begin{align}
T_{\vec{\Delta}_{2L} } | \phi \rangle = \sum_{\bm{k}}  A_{\vec{\Delta}_{2L} } ( \bm{k} ) \phi_{\bm{k}}^{(V)} |  r_{\vec{\Delta}_{2L} } (\bm{k} )  \rangle^{(U)}
\label{theorem1even}
\end{align}
for an even number of projectors as in (\ref{multprojeven}).  Here,  $ \phi_{\bm{k}}^{(V)}  =  \langle \bm{k} | V^\dagger | \phi   \rangle $ is the amplitude of $ | \phi \rangle $ in the $ V$-basis,  $ L \ge 1$, 
$ r_{\vec{\Delta}} (\bm{k} ) $ is the result of the recursion relation (\ref{recursionresult}) and (\ref{recursionresulteven}) and $ A_{\vec{\Delta}} ( \bm{k} ) $ is defined in (\ref{atotalcoeff}) and (\ref{atotalcoeffeven}).  
\end{theorem}
\begin{proof}
Inserting the resolution of the identity in the $ V$-basis
\begin{align}
T_{\vec{\Delta} }  | \phi \rangle  = T_{\vec{\Delta} } \sum_{\bm{k}}  |\bm{k} \rangle^{(V)}  \langle \bm{k} |^{(V)}  | \phi \rangle 
\end{align}
and using the results (\ref{cor1result}) and (\ref{cor1resulteven})  give  (\ref{theorem1}) and  (\ref{theorem1even}).  
\end{proof}

It follows straightforwardly from (\ref{theorem1}) that the probability of a particular sequence of projections is given by 
\begin{align}
p_{\vec{\Delta}} & = \langle \phi | T_{\vec{\Delta}}^\dagger 
T_{\vec{\Delta} } | \phi \rangle  \nonumber \\
& = \sum_{\bm{k}} | A_{\vec{\Delta}} ( \bm{k} ) \phi_{\bm{k}}^{(V)}|^2 .  
\end{align}
The states  (\ref{theorem1}) and  (\ref{theorem1even}) can then be normalized by simply dividing by $ \sqrt{p_{\vec{\Delta}}} $.  

Equations (\ref{theorem1})  and  (\ref{theorem1even}) allows for a highly efficient way to evaluate the result of a sequence of projections.  In words, the procedure is as follows.  Given an arbitrary initial state $ | \phi \rangle $, first expand the state in the $ V$-basis, defined by (\ref{initialstatevk1k2}).  Then for each state labeled by $ \bm{k} $, we can recursively apply  (\ref{recursivekrelation}) 
until we find $ \bm{k}_L $.  This results in obtaining the sequence of $ \bm{k}_l, \bm{k}_l' $, from which the overall coefficient $ A $ can be found through (\ref{atotalcoeff}).  Multiplying the coefficient by the associated term in the superposition gives (\ref{theorem1}).  This procedure is far more efficient than evaluating $ 2L + 1$ matrix multiplications directly, each with an overhead of $ (N+1)^6 $.

\subsection{Example: stochastic evolution}

\begin{figure}[t]%
\includegraphics[width=\linewidth]{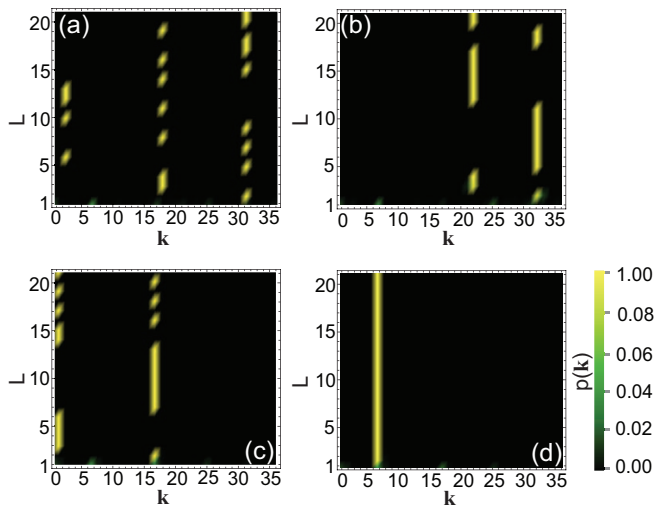}
\caption{The probability distribution (\ref{probpurestates}) of the quantum state after multiple QND projections (\ref{multitpsi}) for the initial  (a),(b) random state; (c),(d) $S_x$-polarized state. For all calculations number of atoms per BEC is taken as $ N = 5 $. States are enumerated by $ [\bm{k} ] = k_2 (N+1) + k_1$. }
\label{fig7}%
\end{figure}

While (\ref{theorem1}) and (\ref{theorem1even}) gives the general result for an arbitrary sequence of projections, for a typical sequence the resulting dynamics is often rather simple. 
Figure \ref{fig7} shows the probability distribution of evolving a random initial state with the sequence of QND projections 
\begin{align}
p( \bm{k}) & =  \frac{|\langle \bm{k} |{^{(V)}} T_{\vec{\Delta}_{2L+1} } | \phi \rangle|^2 }{p_{\vec{\Delta}_{2L+1} } } ,
%\nonumber \\
%& =\frac{1}{p_{\vec{\Delta}_{2L+1} } }  \left|  \langle \bm{k} | V^\dagger  P_{\Delta^z_{L+1}}^{(z)}\prod_{l=1}^L  \Big[ P_{\Delta^x_l}^{(x)}P_{\Delta^z_l}^{(z)}  \Big] | \phi \rangle \right|^2 .
\label{probpurestates}
\end{align}
where we consider an odd number of projections, such that the output state is in the $ V $-basis. Here the measurement outcomes $ \vec{\Delta}_{2L+1} $ are chosen randomly according to their measurement probabilities.  We see that the state quickly becomes dominated by a single state in the $ V$-basis $ | \bm{k} \rangle^{(V)} $.  Depending upon the measurement outcomes and the initial state, the state can remain fixed in a particular state (Fig. \ref{fig7}(d)), or jump randomly between several values (Fig. \ref{fig7}(a)-(c)).  This depends upon the recursion relations (\ref{recursivekrelation}).  However, at any given step in the projection sequence, the state is typically in only only of the $ | \bm{k} \rangle^{(V)} $ states.  

The reason for this can be seen from (\ref{theorem1}). After a sequence of projections is applied, the state becomes modified by the factor $  A_{\vec{\Delta}} ( \bm{k} ) $.  Since this is a product of matrix elements according to (\ref{atotalcoeff}), depending upon the particular $ \vec{\Delta} $ that occurs, one of the coefficients of $ | \bm{k} \rangle^{(V)} $ becomes dominant, due to the exponential increase of a dominant factor.  Similar results occur for an even number of projections, with the main difference being that the final state is in the $ U $-basis.  

We summarize the effect of multiple QND projections by the following basic rule-of-thumb: {\it it collapses the state to one of the states in the $ V$-basis for an odd number of projections, and one of the $ U$-basis states for an even number of projections}.

\section{Sequential QND projections: Mixed state analysis}
\label{sec:mixedstate}

In the previous section we considered the effect of applying a sequence of projections $ T_{\vec{\Delta}} $ to a pure state.  This results in a stochastic evolution of the state, as illustrated in Fig. \ref{fig7}. This corresponds to a single run of a particular experiment.  It is also useful to analyze this in a statistical sense, where we deal with a probabilistic evolution of the state.  The natural language for this are density matrices, and in this section we adapt the results of the previous section to the mixed state case.

\subsection{Sequence of QND projections on a mixed state}

Let us first generalize the results of the previous section to the case when the initial state is a mixed state with density matrix $ \rho_0$.  After performing a projection sequence the resulting unnormalized density matrix is
\begin{align}
 \tilde{\rho}_{\vec{\Delta}} = T_{\vec{\Delta}} \rho_0 T_{\vec{\Delta}}^\dagger
\end{align}
occurring with probability 
\begin{align}
p_{\vec{\Delta}} = \text{Tr} ( T_{\vec{\Delta}}^\dagger T_{\vec{\Delta}} \rho_0 )  .  
\end{align}
Instead of a single shot outcome as we considered in the last section, we now wish to obtain the average over many possible runs of the QND projector sequence.  Let us here restrict ourselves to the case that there is an odd number of projections $ \vec{\Delta} = \vec{\Delta}_{2L+1}$.  Averaging over all possible outcomes the density matrix is
\begin{align}
\rho^{(V)} & =\sum_{\vec{\Delta} } (\prod_{l=1}^L\Lambda_{\Delta_l^x\Delta_{l+1}^z}^{\dagger}\Lambda_{\Delta_l^x\Delta_l^z})\rho^{(V)}_0 (\prod_{l=1}^L\Lambda_{\Delta_l^x\Delta_{l+1}^z}^{\dagger}\Lambda_{\Delta_l^x\Delta_l^z})^{\dagger}
\label{avgrho}
\end{align}
where we used (\ref{multprojsimpl}) and we defined density matrices rotated into the $ V$-basis as
\begin{align}
    \rho^{(V)} = V^\dagger \rho V .  
\end{align}
Alternatively, we may use (\ref{theorem1}) to write 
\begin{align}
\rho^{(V)} & =\sum_{\vec{\Delta}} \sum_{\bm{k} \bm{k}'} 
A_{\vec{\Delta}} ( \bm{k}) A_{\vec{\Delta}}^* ( \bm{k} ') 
\rho_{\bm{k} \bm{k}'}^{(V)} | r_{\vec{\Delta}} (\bm{k}) \rangle \langle r_{\vec{\Delta}} (\bm{k}') |
\label{theorem1densitymat}
\end{align}
where $ \rho_{\bm{k} \bm{k}'}^{(V)} = \langle \bm{k} | V^\dagger \rho_0 V |  \bm{k}' \rangle $.

\subsection{Example: probabilistic evolution}

We now numerically evaluate the density matrix evolution to illustrate the effect of the multiple projections.  To evaluate the density matrix, rather than the expressions (\ref{avgrho}) or (\ref{theorem1densitymat}) 
which contain a large number of summations over $ \vec{\Delta} $, it is more efficient to evaluate the density matrices iteratively. 
Consider the application of an even number of projectors (\ref{multprojeven}) and averaging over all outcomes we have
\begin{align}
\rho_l' = \sum_{\Delta^z_1, \Delta^x_1, \dots, \Delta^x_l} P^{(x)}_{\Delta^x_l} \dots P^{(z)}_{\Delta^z_1} \rho_0 
P^{(z)}_{\Delta^x_1} \dots P^{(x)}_{\Delta^z_l} . 
\end{align}
For an odd number of projectors we have instead
\begin{align}
\rho_{l+1} = \sum_{\Delta^z_1, \Delta^x_1, \dots, \Delta^x_l,\Delta^z_{l+1} } P^{(z)}_{\Delta^z_{l+1}} P^{(x)}_{\Delta^x_l} \dots P^{(z)}_{\Delta^z_1} \rho_0 
P^{(z)}_{\Delta^x_1} \dots P^{(x)}_{\Delta^z_l}  P^{(z)}_{\Delta^z_{l+1}} . 
\end{align}
Then using  property (\ref{idemortho}) of projection operators, we may relate the two density matrices by
\begin{align}
\rho_{l+1} = \sum_{\Delta^x, \Delta^z } T^\dagger_{\Delta^z \Delta^x} \rho_l' T_{\Delta^z \Delta^x} . 
\label{rhodashtorho}
\end{align}
Similarly, we may transform a density matrix with an odd number of projectors to an even one according to
\begin{align}
\rho_{l}' = \sum_{\Delta^x, \Delta^z } T_{\Delta^z \Delta^x} \rho_l T^\dagger_{\Delta^z \Delta^x} . 
\label{rhotorhodash}
\end{align}

Applying (\ref{tsvdform}) we then find that
\begin{align}
{\rho_l'}^{(U)} & = \sum_{\Delta^x, \Delta^z } \Lambda_{\Delta^x\Delta^z}^\dagger {\rho_l}^{(V)} \Lambda_{\Delta^x\Delta^z}    \label{vtoudensitymat}  \\
    \rho_{l+1}^{(V)} & = \sum_{\Delta^x, \Delta^z } \Lambda_{\Delta^x\Delta^z} {\rho_l'}^{(U)} \Lambda_{\Delta^x\Delta^z}^\dagger   .
    \label{utovdensitymat}
\end{align}
where we defined
\begin{align}
\rho^{(U)} = U^\dagger \rho U .  
\end{align}
From (\ref{vtoudensitymat}) and (\ref{utovdensitymat}) we may find the resulting density matrix iteratively, by starting with the initial state $ \rho_0 = \rho_1 $.

%From (\ref{avgrho}) we see that each application of a projector
%
%\begin{align}
%\rho_{n+1}^{(V)} = \left\{
%\begin{array}{ll}
%\sum_{\Delta^x \Delta^z} \Lambda_{\Delta^x\Delta^z} \rho^{(V)}_n \Lambda_{\Delta_l^x\Delta_l^z}^{\dagger}  \hspace{1cm} & (n \text{ even}) \nonumber \\
%\sum_{\Delta^x \Delta^z} \Lambda_{\Delta^x\Delta^z} ^\dagger  \rho^{(V)}_n \Lambda_{\Delta_l^x\Delta_l^z} \hspace{1cm} & (n \text{ odd}) 
%\end{array}
%\right.
%\end{align}
%
%Here $ n $ counts the applications of the $  \Lambda$ matrices, and an equivalent number of projection sequences to (\ref{avgrho}) would be to start at $ n = 0 $ and end at $ n = 2L $.  

Figure \ref{fig8} shows the density matrix evolution starting from two initial states, the state (\ref{initialstatexx}) and a random pure state. In Fig. \ref{fig8}(a)(b) we show the diagonal elements of the density matrix in the $ V $-basis as a function of the total number of rounds of projections $ L $.  We see that in both cases the states rapidly converge to a fixed probability distribution.  The particular distribution that is obtained, depends upon the initial state.  The final convergent probability distributions correspond to the proportions of the state that would be obtained after many runs of the stochastic evolution as shown in Fig. \ref{fig7}.

\begin{figure}[t]%
\includegraphics[width=\linewidth]{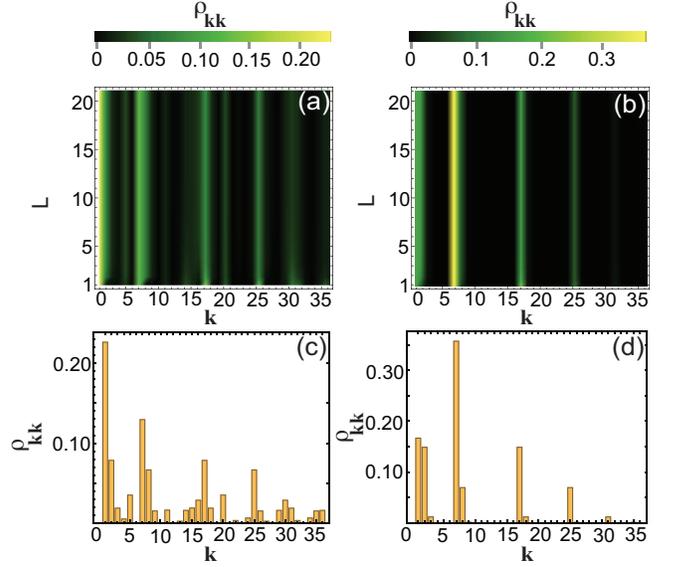}
\caption{The probability distribution of the quantum mixed state after multiple QND projections (\ref{multitpsi}) for the initial  (a) random state; (b) $S_x$-polarized state. For all calculations number of atoms per BEC is taken as $ N = 5 $. The corresponding probability bar charts are also shown at the bottom for each case. States are enumerated by $ [\bm{k} ] = k_2 (N+1) + k_1$.
\label{fig8} }
\end{figure}

\subsection{Steady-state distribution}

We now show a method for finding the steady-state distribution of the density matrix, as illustrated in Fig. \ref{fig8}.  After a large number of iterations  $ L $, the density matrix remains unchanged and reaches a fixed point.  We again consider the case where there is an odd number of projectors, such as the final density matrix is in the $ V$-basis.  Combining (\ref{rhodashtorho}) and (\ref{rhotorhodash}), we then demand that
\begin{align}
\rho = \sum_{\vec{\Delta}_3} T_{\vec{\Delta}_3} \rho T_{\vec{\Delta}_3}^\dagger
\label{steadystatecondition}
\end{align}
where $ \vec{\Delta}_3 = ( \Delta^z, \Delta^x, {\Delta^z}' )$. Expanding the density matrix in the $ V $-basis, we have
\begin{align}
\rho = \sum_{\bm{k}, \bm{k}'} \rho_{\bm{k} \bm{k}'}^{(V)} | \bm{k} \rangle^{(V)} \langle  \bm{k}'  |^{(V)}  .  
\end{align}
Using (\ref{cor1result}), the steady-state relation can be written
\begin{align}
\rho = \sum_{\vec{\Delta}_3} \sum_{\bm{k}, \bm{k}'}  A_{\vec{\Delta}_3} (\bm{k}) A_{\vec{\Delta}_3}^* (\bm{k}') \rho_{\bm{k} \bm{k}'}^{(V)}  
| r_{\vec{\Delta}_3} (\bm{k}) \rangle \langle r_{\vec{\Delta}_3} (\bm{k}') |^{(V)} .
\label{steadystaterel2}
\end{align}

We wish to obtain the distribution of the diagonal elements of the density matrix at steady state.  Define a vector consisting of the diagonal elements in the $ V $-basis
\begin{align}
d_{\bm{k}} = \langle \bm{k} | \rho | \bm{k} \rangle^{(V)}  .  
\end{align}
Taking the diagonal matrix elements of (\ref{steadystaterel2}) we obtain the relation
\begin{align}
d_{\bm{m} } = \sum_{\bm{k}} \sum_{\vec{\Delta}_3} | A_{\vec{\Delta}_3} (\bm{k}) |^2 d_{\bm{k}} , 
\label{funnymatrix}
\end{align}
where $\bm{m} = r_{\vec{\Delta}_3} (\bm{k})$. The sum over $ \bm{k}' $ in (\ref{steadystaterel2}) collapses to $\bm{k}'= \bm{k} $ since this is the only way to satisfy $ r_{\vec{\Delta}_3} (\bm{k}) = r_{\vec{\Delta}_3} (\bm{k}') $, where $ r $ is a permutation operation.  If we define a $ (N+1)^2 \times (N+1)^2  $ matrix  with elements
\begin{align}
{\cal A}_{\bm{m} \bm{k} } =  \sum_{\vec{\Delta}_3} | A_{\vec{\Delta}_3} (\bm{k}) |^2 
\end{align}
with  $\bm{m} = r_{\vec{\Delta}_3} (\bm{k})$ again, then (\ref{funnymatrix}) can be written as a matrix multiplication with elements
\begin{align}
d_{\bm{m} } = \sum_{\bm{k}} {\cal A}_{\bm{m} \bm{k} }  d_{\bm{k}} . 
\label{funnymatrix2}
\end{align}
This is an eigenvalue equation, where $d$ must be an eigenvector of $ \cal A$ with eigenvalue 1.  

We have numerically verified that the same probability distributions are obtained using (\ref{funnymatrix2}) as the iterative method as calculated in Fig. \ref{fig8}.

\begin{figure}[t]%
\includegraphics[width=\linewidth]{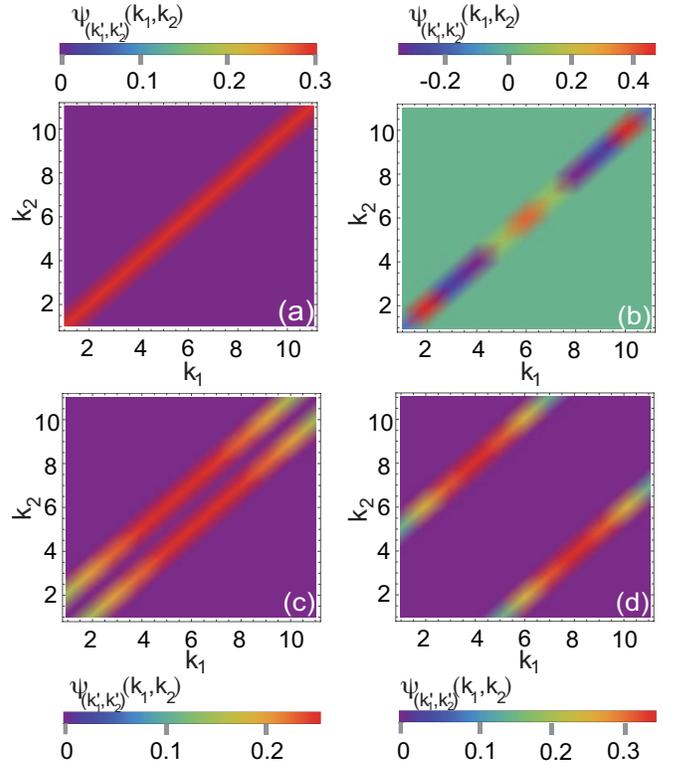}
\caption{The amplitudes of $ V$-basis states as defined by (\ref{vbasisstates}).  (a) $ k_1' = k_2' = 0 $ ($ \Delta = 0 $ sector);  (b) $ k_1' =0,  k_2' = 3 $ ($ \Delta = 0 $ sector); (c) $ k_1' = N + 1,  k_2' = 0 $ ($ \Delta = 1 $ sector); (d) $ k_1' = 7N -5,  k_2' = 0 $ ($ \Delta = 3 $ sector).  We use $ N = 10 $ for our calculations. 
\label{fig9} }
\end{figure}

\section{Properties of the $V$-basis states}
\label{vbasistates}

We have seen in Sec. \ref{sec:purestate} and Sec. \ref{sec:mixedstate} that the effect of a long sequence of projections $ T_{\vec{\Delta}} $ is to collapse an initial state onto one of the $ V $-basis states.  We examined a particular case in Sec. \ref{sec:maxentangle} where the state converges towards a maximally entangled state.  In this section, we examine the nature of the remaining $ V$-basis states.  We again limit our analysis to an odd number of projections.  
\begin{figure}[t]%
\includegraphics[width=\linewidth]{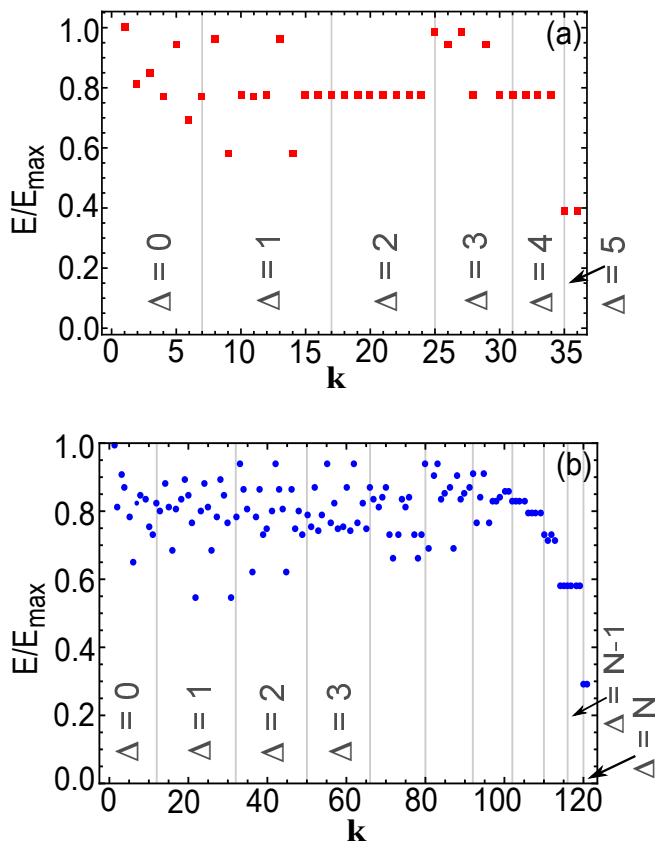}
\caption{Entanglement of the $ V$-basis states as defined by (\ref{entanglevonneumann}) for (a) $N = 5$, (b) $N = 10$.  States are enumerated by $ [\bm{k} ] = k_2 (N+1) + k_1$. Different states in each $\Delta$ sector are separated by vertical line. We have normalized with respect to the maximum entanglement between two BECs, $ E_\text{max} = \log_2 (N+1)$.
\label{fig10} }
\end{figure}
Figure \ref{fig9} shows a gallery of states in the $ V $-basis. We plot the coefficients
\begin{align}
    \psi_{k_1' k_2'} (k_1, k_2)   = \langle k_1, k_2 |  k_1', k_2' \rangle^{(V)}
    \label{vbasisstates}
\end{align}
for various choices of $  k_1', k_2' $. We find that all states are diagonally correlated in $ k_1, k_2$.  This is guaranteed from the fact that last projector $ P_{\Delta^z_{L+1}}^{(z)} $ in the sequence is in the $ z$-basis, and consists of correlated states according to (\ref{projzbasis}).  The $V$-basis states can be categorized into different sectors according to the $\Delta =  \Delta^z_{L+1} $ of the final  projection.  The number of states in each sector is 
\begin{align}
n(\Delta) =\frac{2(N+1-\Delta)}{2^{\delta_{\Delta}}}
\end{align}
which is the same as the rank of each projector in (\ref{projzbasis}).  
The states in each sector have different amplitude relations to ensure orthogonality (Fig. \ref{fig9}(a)(b)).

The diagonal nature of the states ensure that all the $ V$-basis states are entangled. The degree of entanglement however depends upon which of the states $ k_1',  k_2' $ is obtained.  Fig. \ref{fig10} shows the entanglement of the states quantified by the von Neumann entropy (\ref{entanglevonneumann}) that is normalized with respect to the maximum entanglement between two BECs.

The reduced density matrix on BEC 1 is 
\begin{align}
\rho = \text{Tr}_2 |  k_1', k_2' \rangle^{(V)}  \langle  k_1', k_2' |^{(V)} .  
\end{align}
From Fig. \ref{fig10}, we see that all states are entangled to various degrees. The maximally entangled state $ E = E_{\text{max}} $ corresponds to the EPR state (\ref{spineprz}).  The lowest entanglement is achieved in the $ \Delta = N $ sector, corresponding to $ E = 1$.  This corresponds to the states
\begin{align}
   \frac{1}{\sqrt{2}}(  | N, 0 \rangle \pm | 0, N \rangle )
\end{align}
which are NOON states. Since all the $ V$-basis states are entangled, the sequential projection will drive an arbitrary state into an entangled state.

\section{Summary and Conclusions}
\label{sec5}

We have developed a theory of entangling QND measurements between atomic ensembles based on measurement operators.  We first formulated the theory as developed in Ref. \cite{aristizabal2021quantum} in terms of POVMs. We then examined a limiting case of the theory with $ |\alpha \tau |^2 > 1 $ where two POVMs reduce to projection operators. This clarifies the way in which entanglement is generated in the system, but also reveals why non-maximal entanglement is generated with a single QND projection (see for e.g.  (\ref{projexampleonez})).  We then proceeded to analyze the multiple projection case, where a sequence of projections in alternating $ z $ and $ x$ bases are made.  Using a key mathematical observation (Lemma 1), the multiple projections can be simplified greatly.  This resulted in Theorem 1, where it was found that it is possible to evaluate an arbitrary sequence of projections very efficiently, using only a set of recursion relations of the state labels and associated amplitudes.  The same situation can be analyzed in a statistical sense using density matrices, where the convergent state can be obtained for a large number of projections.

The main difference between our results and past works is to firstly formulate the QND entangling operation as an exact POVM or projection operator.  Past methods rely on using various approximation to the spins (e.g. a Holstein-Primakoff approximation). Since our theory is based upon Ref. \cite{aristizabal2021quantum} which treats the spin dynamics exactly, 
the POVM formulation is an exact solution of the QND dynamics.  The projection 
operator is a particular limiting case of this theory, but  should be very accurate for the large photon regime that experiments are typically carried out in. Secondly, we have shown that a sequence of multiple QND measurements (i.e. the stroboscopic measurements) has a very simple mathematical structure that can be exploited.  This yields a method to obtain the state after a number of measurements in a very efficient way.  Normally, the projection sequence requires a calculation that scales as $ (2L+1) (N+1)^6  $, and there is little intuition to the procedure.  We have shown that there is a particular basis in which the states collapse to (the $ U $ and $ V$-basis), and makes clear the nature of the state that is obtained in the procedure. 

The projected sequence to prepare a maximally entangled state as shown in Sec. \ref{sec:maxentangle} is very attractive as a fundamental state preparation protocol for use in quantum information applications. It allows for a way of enhancing the spin correlations and can generate a highly entangled state of atomic ensembles.  A prime application is spinor quantum computing \cite{byrnes2012macroscopic,byrnes2015macroscopic}, where one deals with a many-body wavefunction to encode the logical states, and a regime away from the conventional
Holstein-Primakoff approximated states are used. Hence it is important to know 
the effect of such entangling operations such that they can be used for quantum information applications.  However, it should be pointed out that this is not deterministic, and represents merely one particular measurement outcome. This is similar to the one-ensemble counterpart as experimentally demonstrated in Ref. \cite{behbood2014} where post-selection was used to target the entangled singlet state.  Due to the large Hilbert space of atomic ensembles, a more desirable scheme would be deterministic preparation of a target entangled state. For the one-ensemble case, feedback approaches were examined to achieve this \cite{behbood2013feedback}. We leave investigation of such deterministic schemes as future work.

\begin{acknowledgments}
The authors thank Valerio Proietti for insightful discussions. This work is supported by the    	National Natural Science Foundation of China (NSFC Grant No. 62071301), State Council of the 	People’s Republic of China (Grant No. D1210036A), NSFC Research Fund for International Young Scientists (Grant No. 11850410426), NYU-ECNU Institute of Physics at NYU Shanghai, the 	Science and Technology Commission of Shanghai Municipality (Grant No. 19XD1423000), the China Science and Technology Exchange Center (Grant No. NGA-16-001), and the NYU Shanghai Boost Fund.
\end{acknowledgments}

\appendix

\section{Proof of Lemma 1}
\label{app:proof}

In this section we show that the commutative property (\ref{commuteq}) holds.  First note that for $ \Delta^z \ne {\Delta^z}' $
\begin{align}
R_{\Delta^z \Delta^x} R_{{\Delta^z}' {\Delta^x}'} = 
P_{\Delta^z}^{(z)}  P_{{\Delta^x}}^{(x)} P_{\Delta^z}^{(z)}  
P_{{\Delta^z}'}^{(z)}  P_{{\Delta^x}'}^{(x)} P_{{\Delta^z}'}^{(z)} = 0 
\end{align}
from the orthogonality of projection matrices.  Thus for $ \Delta^z \ne {\Delta^z}' $ the commutative relation (\ref{commuteq}) holds trivially.  Thus we only need to prove (\ref{commuteq}) for the $ \Delta^z = {\Delta^z}' $ case, and our task will be to prove the equivalent relation 
\begin{align}
P_{\Delta^z}^{(z)}  P_{{\Delta^x}'}^{(x)} P_{\Delta^z}^{(z)}  P_{{\Delta^x}''}^{(x)} P_{\Delta^z}^{(z)} {=} P_{\Delta^z}^{(z)} P_{{\Delta^x}''}^{(x)} P_{\Delta^z}^{(z)}  P_{{\Delta^x}'}^{(x)} P_{\Delta^x}^{(x)} .  
\label{thefiveps}
\end{align}

Next let us write the relation (\ref{thefiveps}) in an alternative form  by defining new projection operators 
\begin{align}
  \Pi_{\delta}^{(l)} = \sum_{k}|k,k+\delta\rangle \langle k,k+\delta|^{(l)} 
  \label{piprojdef}
\end{align}
where $ \delta $ can be a positive or negative integer, and $ l \in\{x,z \} $ indicates the basis.  Then our original projection operators (\ref{projzbasis}) are 
\begin{align}
 P_{\Delta}^{(l)} = \Pi_{\Delta}^{(l)} + \Pi_{-\Delta}^{(l)} = \sum_{\sigma=\pm 1}\Pi_{\sigma\Delta}^{(l)} .  
  \label{alternateprojzbasis}
\end{align}
The relation that we wish to prove (\ref{thefiveps}) can be rewritten in terms of the new projection operators as
\begin{align}
& \sum_{\sigma_1 \sigma_2 \sigma_3 \sigma_4  \sigma_5}
\Big( \Pi_{{\sigma_1 \Delta^z}}^{(z)}  \Pi_{{\sigma_2 \Delta^x}'}^{(x)} \Pi_{\sigma_3\Delta^z}^{(z)}  \Pi_{{\sigma_4 \Delta^x}''}^{(x)} \Pi_{{\sigma_5 \Delta^z}}^{(z)}  \nonumber \\
& - \Pi_{{\sigma_1 \Delta^z}}^{(z)}   \Pi_{{\sigma_4 \Delta^x}''}^{(x)} \Pi_{\sigma_3\Delta^z}^{(z)}  \Pi_{{\sigma_2 \Delta^x}'}^{(x)} \Pi_{{\sigma_5 \Delta^z}}^{(z)}   \Big) = 0 .  
\label{projequiv1}
\end{align}
it is possible to relabel the $ \sigma_i$ variables since they are dummy indices, and we have made a convenient choice for the purposes of our proof.  
We will show that the quantity inside parentheses is always zero, proving the desired relation. 
%If this is true, then a sum over zeros is zero, and the desired relation is proven. [AP: not a sum over zeros, since the terms are non-zero when sigma1 and sigma5 differ, but (should) cancel with the terms where sigma5 and sigma1 swap. Terms where sigma1=sigma5, however, are each zero. Tim demonstrated both these phenomena with numerical results]

The key insight to proving the desired relation is to notice that the
projectors (\ref{piprojdef}) project the two BECs onto an eigenstate of $S^z_1-S^z_2$, such that the spin difference between the BECs is fixed. It will be convenient to transform the projectors such that instead of a spin difference, the projectors (\ref{piprojdef}) project on to the total spin $S^z_1+S^z_2$.  This can be achieved by flipping the second spin such that the transformed state is
\begin{align}
   e^{iS^y_2\frac{\pi}{2}} |k\rangle = (-1)^k|N-k\rangle . 
   \label{spinflip}
\end{align}
We can therefore define the rotated projection operator
\begin{align}
 \tilde{\Pi}_{\Delta}^{(z)}  & =  e^{iS^y_2\frac{\pi}{2}} \Pi_{\Delta}^{(z)} e^{-iS^y_2\frac{\pi}{2}}  \nonumber \\
 & = \sum_{k}|k,N-k-\Delta\rangle \langle k,N-k-\Delta|,
 \label{pirotation}
\end{align}
which consists of states with total spin $S^z_1 + S^z_2 = -2 \Delta $. 

Let us now change the notation of the Fock states such that they are instead interpreted as angular momentum eigenstates.  For a single BEC, we may write \cite{timquantumoptics2020}
\begin{align}
|k\rangle = \Big|j=\frac{N}{2}, m = k-\frac{N}{2}\Big\rangle
    \label{angmomentumbasis}
\end{align}
where `$j$' is the angular momentum quantum number and `$m$' is the quantum number associated with the spin along the $ z$-direction.  Then the rotated projectors can be written in this notation
\begin{align}
 \tilde{\Pi}_{\Delta}^{(z)} = \sum_{m} & |j_1=\frac{N}{2} ,m;j_2=\frac{N}{2} ,-m-\Delta\rangle \nonumber \\
  &  \langle j_1=\frac{N}{2},m;j_2=\frac{N}{2},-m-\Delta| .  
   \label{projnewbasis}
\end{align}

Now we may rewrite the states in (\ref{projnewbasis}) in terms of the total angular momentum basis  $|J,M\rangle$, consisting of two individual spins of equal angular momenta $ j_1 = j_2 = \frac{N}{2} $. The total spin operator is defined
\begin{align}
\bm{J} = \bm{j}_1 + \bm{j}_2 = \frac{\bm{S}_1 + \bm{S}_2}{2}  .  
\end{align}
The conversion between the bases is achieved by Clebsch-Gordan coefficients
\begin{align}
|j_1,m_1;j_2,m_2\rangle = & \sum_{J=|j_1 - j_2|}^{j_1 +j_2} | J, M = m_1 + m_2 \rangle  \nonumber \\
& \times \langle J, M = m_1 + m_2 | j_1,m_1;j_2,m_2\rangle .  
    \label{prodtocouple}
\end{align}
Substituting this into (\ref{projnewbasis}), we obtain
\begin{align}
      \tilde{\Pi}_{\Delta}^{(z)}  = \sum_{J}|J,-\Delta\rangle^{(z)} \langle J,-\Delta|^{(z)} . 
      \label{rotpiz}
\end{align}
Here we used the $ m$-sum unitarity relation for Clebsch-Gordan coefficients \cite{thompson2008angular},
%
%\begin{align}
%\sum_J  \langle  j_1,m_1';j_2,m_2' | J, M \rangle \langle J, M | j_1,m_1;j_2,m_2\rangle = \delta_{m_1 m_1'} \delta_{m_2 m_2'} 
%\end{align}
%
\begin{align}
\sum_{m_1,m_2}& \langle J',M'|j_1,m_1;j_2,m_2\rangle\langle j_1,m_1;j_2,m_2|J,M\rangle
\nonumber \\
     & =\delta_{JJ'}\delta_{MM'} .  
     \label{unitarysum}
\end{align}
The sum in (\ref{rotpiz}) is over all total angular momentum states that have the same $ J^z $ eigenvalue.
Similarly, we have the projectors defined in the $ x $-basis
\begin{align}
      \tilde{\Pi}_{\Delta^x}^{(x)}
      & =  e^{-iJ^y\frac{\pi}{4}}   \tilde{\Pi}_{\Delta}^{(z)}   e^{iJ^y\frac{\pi}{4}}  \\ \nonumber
       & = \sum_{J}|J,-\Delta\rangle^{(x)} \langle J,-\Delta|^{(x)} 
\end{align}

Let us now evaluate the sequence of projectors
\begin{align}
& \tilde{\Pi}_{\Delta}^{(z)} \tilde{ \Pi}_{{\Delta_2}}^{(x)} \tilde{\Pi}_{\Delta_3}^{(z)}   \tilde{\Pi}_{{\Delta_4}}^{(x)} \tilde{\Pi}_{\Delta}^{(z)} 
 = \sum_J |J,-\Delta\rangle^{(z)} \langle J, -\Delta  |^{(z)} \nonumber \\
& \times d_{-\Delta, -\Delta_2}^J  
d_{ -\Delta_3,-\Delta_2}^J d_{ -\Delta_3,-\Delta_4}^J d_{ -\Delta,-\Delta_4}^J , 
\label{piseq1}
\end{align}
where we used the fact that a rotation of the angular momentum states preserves $ J $  \cite{thompson2008angular}
\begin{align}
    \text{}^{(z)}\langle J,M|J',M'\rangle^{(x)} &=\text{}^{(z)}\langle J,M|e^{-iJ^y \frac{\pi}{4}} |J',M'\rangle^{(z)} \nonumber \\ &= \delta_{J,J'}d^{J}_{M,M'}  ,  
\end{align}
and we defined
\begin{align}
  d^{J}_{M,M'} =   \text{}^{(z)}\langle J,M|J,M'\rangle^{(x)} .  
\end{align}
We may also evaluate the sequence with the ordering of the $ \Delta_2 $ and $ \Delta_4 $ interchanged
\begin{align}
& \tilde{\Pi}_{\Delta}^{(z)} \tilde{ \Pi}_{{\Delta_4}}^{(x)} \tilde{\Pi}_{\Delta_3}^{(z)}   \tilde{\Pi}_{{\Delta_2}}^{(x)} \tilde{\Pi}_{\Delta}^{(z)} 
 = \sum_J |J,-\Delta\rangle^{(z)} \langle J, -\Delta |^{(z)} \nonumber \\
& \times d_{-\Delta, -\Delta_4}^J  
d_{ -\Delta_3,-\Delta_4}^J d_{ -\Delta_3,-\Delta_2}^J d_{ -\Delta,-\Delta_2}^J .  
\label{piseq2}
\end{align}
The right hand sides of (\ref{piseq1}) and (\ref{piseq2}) are the same, hence we establish that 
\begin{align}
 \tilde{\Pi}_{\Delta}^{(z)} \tilde{ \Pi}_{{\Delta_2}}^{(x)} \tilde{\Pi}_{\Delta_3}^{(z)}   \tilde{\Pi}_{{\Delta_4}}^{(x)} \tilde{\Pi}_{\Delta}^{(z)}  = \tilde{\Pi}_{\Delta}^{(z)} \tilde{ \Pi}_{{\Delta_4}}^{(x)} \tilde{\Pi}_{\Delta_3}^{(z)}   \tilde{\Pi}_{{\Delta_2}}^{(x)} \tilde{\Pi}_{\Delta}^{(z)} .  
  \label{piresult1}
\end{align}

Now let us consider the sequence
\begin{align}
& \tilde{\Pi}_{\Delta}^{(z)} \tilde{ \Pi}_{{\Delta_2}}^{(x)} \tilde{\Pi}_{\Delta_3}^{(z)}   \tilde{\Pi}_{{\Delta_4}}^{(x)} \tilde{\Pi}_{-\Delta}^{(z)} 
 = \sum_J |J,-\Delta\rangle^{(z)} \langle J, \Delta  |^{(z)} \nonumber \\
& \times d_{-\Delta, -\Delta_2}^J  
d_{ -\Delta_3,-\Delta_2}^J d_{ -\Delta_3,-\Delta_4}^J d_{ \Delta,-\Delta_4}^J , 
\label{piseq3}
\end{align}
The sequence with $ \Delta_2 $ and $ \Delta_4 $ interchanged is 
\begin{align}
& \tilde{\Pi}_{\Delta}^{(z)} \tilde{ \Pi}_{{\Delta_4}}^{(x)} \tilde{\Pi}_{\Delta_3}^{(z)}   \tilde{\Pi}_{{\Delta_2}}^{(x)} \tilde{\Pi}_{-\Delta}^{(z)} 
 = \sum_J |J,-\Delta\rangle^{(z)} \langle J, \Delta  \rangle^{(z)} \nonumber \\
& \times d_{-\Delta, -\Delta_4}^J  
d_{ -\Delta_3,-\Delta_4}^J d_{ -\Delta_3,-\Delta_2}^J d_{ \Delta,-\Delta_2}^J .  
\label{piseq4}
\end{align}
Using the identity \cite{thompson2008angular}
\begin{align}
    \text{}^{(z)}\langle J,-M|J,M'\rangle^{(x)} &=  \text{}^{(z)}\langle J,M|J,M'\rangle^{(x)}  (-1)^{J + M}
\end{align}
we have
\begin{align}
d_{-M, M'}^J = (-1)^{J+M} d_{M, M'}^J .
\end{align}
Using this in relation we may equate the right hand sides of (\ref{piseq3}) and (\ref{piseq4}) and we establish that 
\begin{align}
 \tilde{\Pi}_{\Delta}^{(z)} \tilde{ \Pi}_{{\Delta_2}}^{(x)} \tilde{\Pi}_{\Delta_3}^{(z)}   \tilde{\Pi}_{{\Delta_4}}^{(x)} \tilde{\Pi}_{-\Delta}^{(z)}  = \tilde{\Pi}_{\Delta}^{(z)} \tilde{ \Pi}_{{\Delta_4}}^{(x)} \tilde{\Pi}_{\Delta_3}^{(z)}   \tilde{\Pi}_{{\Delta_2}}^{(x)} \tilde{\Pi}_{-\Delta}^{(z)} . 
 \label{piresult2}
\end{align}

We may now write the results (\ref{piresult1}) and (\ref{piresult2}) in an equivalent form, by inverting the transformation (\ref{pirotation}) and removing the tildes from the projectors
\begin{align}
    &\Pi_{\Delta}^{(z)}  \Pi_{{\Delta}''}^{(x)} \Pi_{\Delta}^{(z)}  \Pi_{{\Delta}'}^{(x)} \Pi_{\Delta}^{(z)}   = e^{-iS^y_2\frac{\pi}{2}} \tilde{\Pi}_{\Delta}^{(z)} \tilde{ \Pi}_{{\Delta}''}^{(x)} \tilde{\Pi}_{\Delta}^{(z)}  \tilde{\Pi}_{{\Delta}'}^{(x)} \tilde{\Pi}_{\Delta}^{(z)}  e^{iS^y_2\frac{\pi}{2}} .   \label{Eq:TransformedProjectorExpression}
\end{align}
Combining the two cases (\ref{piresult1}) and (\ref{piresult2}) into one relation we have 
\begin{align}
 \Pi_{\Delta}^{(z)} \Pi_{{\Delta_2}}^{(x)} \Pi_{\Delta_3}^{(z)}   \Pi_{{\Delta_4}}^{(x)} \Pi_{\pm \Delta}^{(z)}  = \Pi_{\Delta}^{(z)} 
\Pi_{{\Delta_4}}^{(x)} \Pi_{\Delta_3}^{(z)}   \Pi_{{\Delta_2}}^{(x)} \Pi_{\pm \Delta}^{(z)} . 
 \label{piresultall}
\end{align}

We may now prove the desired relation (\ref{projequiv1}). Setting $\Delta = \sigma_1 \Delta^z$, $ \Delta_2 = \sigma_2 {\Delta^x} ' $, $ \Delta_3 = \sigma_3 \Delta^z $, $ \Delta_4 = \sigma_4 {\Delta^x}'' $ and applying (\ref{piresultall}), for cases with $ \sigma_1 = \sigma_5$ the content of the brackets in (\ref{projequiv1}) is zero, according to the  $ + $ version of (\ref{piresultall}).  For cases with $ \sigma_1 = -\sigma_5$, we can apply the $ - $ version of (\ref{piresultall}), and the quantity inside the brackets in (\ref{projequiv1}) is again zero. This proves the desired relation.

\section{An example of Lemma 1}
\label{app:examplelemma1}

We provide an explicit example for $ N = 2$.  The rotation matrix (\ref{unitaryrotation}) can be evaluated using Eq. (5.162) and (5.163) in Ref. \cite{timquantumoptics2020}.  The rotation matrix is for this case
\begin{align}
{\cal U}(\frac{\pi}{2},0) = \frac{1}{4}\left(
\begin{array}{ccc}
 U_1 & U_{2} & U_{1} \\
 -U_{2} & U_{3} & U_{2} \\
 U_{1} & -U_{2}  & U_{1} \\
\end{array}
\right)
\end{align}
where we defined 
\begin{align}
U_{1} &= \left(
\begin{array}{ccc}
 1 & -\sqrt{2} & 1\\
 \sqrt{2} & 0 &  -\sqrt{2}\\
 1 & \sqrt{2} & 1 \\
 \end{array}
 \right) \\
 U_{2}  &=  \left(
\begin{array}{ccc}
 -\sqrt{2} & 2 & -\sqrt{2}\\
 -2 & 0 & 2\\
 -\sqrt{2} & -2 & -\sqrt{2} \\
 \end{array}
 \right)\\
 U_{3} & = \left(
\begin{array}{ccc}
 0 & 0 & 0\\
 0 & 0 & 0\\
 0 & 0 & 0 \\
 \end{array}
 \right)  .
\end{align}
This can be used to calculated the sequence of two projectors as in (\ref{tmatfortwo}).  For example, for $ \Delta^z = \Delta^x = 0 $
\begin{align}
T_{00} = \frac{1}{8}\left(
\begin{array}{ccccccccc}
 3 & 0 & 0 & 0 & 2 & 0 & 0 & 0 & 3 \\
 0 & 0 & 0 & 0 & 0 & 0 & 0 & 0 & 0 \\
 -1 & 0 & 0 & 0 & 2 & 0 & 0 & 0 & -1 \\
 0 & 0 & 0 & 0 & 0 & 0 & 0 & 0 & 0 \\
 2 & 0 & 0 & 0 & 4 & 0 & 0 & 0 & 2 \\
 0 & 0 & 0 & 0 & 0 & 0 & 0 & 0 & 0 \\
 -1 & 0 & 0 & 0 & 2 & 0 & 0 & 0 & -1 \\
 0 & 0 & 0 & 0 & 0 & 0 & 0 & 0 & 0 \\
 3 & 0 & 0 & 0 & 2 & 0 & 0 & 0 & 3 \\
\end{array}
\right)
\end{align}
and for $ \Delta^z = \Delta^x = 2 $
\begin{align}
T_{22} = \frac{1}{8}\left(
\begin{array}{ccccccccc}
 0 & 0 & 1 & 0 & 0 & 0 & 1 & 0 & 0 \\
 0 & 0 & 0 & 0 & 0 & 0 & 0 & 0 & 0 \\
 0 & 0 & 1 & 0 & 0 & 0 & 1 & 0 & 0 \\
 0 & 0 & 0 & 0 & 0 & 0 & 0 & 0 & 0 \\
 0 & 0 & -2 & 0 & 0 & 0 & -2 & 0 & 0 \\
 0 & 0 & 0 & 0 & 0 & 0 & 0 & 0 & 0 \\
 0 & 0 & 1 & 0 & 0 & 0 & 1 & 0 & 0 \\
 0 & 0 & 0 & 0 & 0 & 0 & 0 & 0 & 0 \\
 0 & 0 & 1 & 0 & 0 & 0 & 1 & 0 & 0 \\
\end{array}
\right) .  
\end{align}
The common unitary operators in (\ref{tsvdform}) are 
\begin{align}
U = \left(
\begin{array}{ccccccccc}
 \frac{1}{\sqrt{3}} & \frac{1}{2 \sqrt{6}} & -\frac{1}{\sqrt{2}} & \frac{1}{2 \sqrt{2}} & 0 & 0 & 0 & 0 & 0 \\
 0 & 0 & 0 & 0 & \frac{1}{2} & -\frac{1}{\sqrt{2}} & 0 & \frac{1}{2} & 0 \\
 0 & -\frac{\sqrt{\frac{3}{2}}}{2} & 0 & \frac{1}{2 \sqrt{2}} & 0 & 0 & 0 & 0 & -\frac{1}{\sqrt{2}} \\
 0 & 0 & 0 & 0 & \frac{1}{2} & 0 & -\frac{1}{\sqrt{2}} & -\frac{1}{2} & 0 \\
 \frac{1}{\sqrt{3}} & -\frac{1}{\sqrt{6}} & 0 & -\frac{1}{\sqrt{2}} & 0 & 0 & 0 & 0 & 0 \\
 0 & 0 & 0 & 0 & \frac{1}{2} & 0 & \frac{1}{\sqrt{2}} & -\frac{1}{2} & 0 \\
 0 & -\frac{\sqrt{\frac{3}{2}}}{2} & 0 & \frac{1}{2 \sqrt{2}} & 0 & 0 & 0 & 0 & \frac{1}{\sqrt{2}} \\
 0 & 0 & 0 & 0 & \frac{1}{2} & \frac{1}{\sqrt{2}} & 0 & \frac{1}{2} & 0 \\
 \frac{1}{\sqrt{3}} & \frac{1}{2 \sqrt{6}} & \frac{1}{\sqrt{2}} & \frac{1}{2 \sqrt{2}} & 0 & 0 & 0 & 0 & 0 \\
\end{array}
\right)
\end{align}
and 
\begin{align}
    V = \left(
\begin{array}{ccccccccc}
 \frac{1}{\sqrt{3}} & \frac{1}{\sqrt{6}} & -\frac{1}{\sqrt{2}} & 0 & 0 & 0 & 0 & 0 & 0 \\
 0 & 0 & 0 & \frac{1}{2} & -\frac{1}{\sqrt{2}} & 0 & \frac{1}{2} & 0 & 0 \\
 0 & 0 & 0 & 0 & 0 & 0 & 0 & \frac{1}{\sqrt{2}} & -\frac{1}{\sqrt{2}} \\
 0 & 0 & 0 & \frac{1}{2} & 0 & -\frac{1}{\sqrt{2}} & -\frac{1}{2} & 0 & 0 \\
 \frac{1}{\sqrt{3}} & -\sqrt{\frac{2}{3}} & 0 & 0 & 0 & 0 & 0 & 0 & 0 \\
 0 & 0 & 0 & \frac{1}{2} & 0 & \frac{1}{\sqrt{2}} & -\frac{1}{2} & 0 & 0 \\
 0 & 0 & 0 & 0 & 0 & 0 & 0 & \frac{1}{\sqrt{2}} & \frac{1}{\sqrt{2}} \\
 0 & 0 & 0 & \frac{1}{2} & \frac{1}{\sqrt{2}} & 0 & \frac{1}{2} & 0 & 0 \\
 \frac{1}{\sqrt{3}} & \frac{1}{\sqrt{6}} & \frac{1}{\sqrt{2}} & 0 & 0 & 0 & 0 & 0 & 0 \\
\end{array}
\right)    . 
\end{align}
The associated singular matrices are 
\begin{align}
\Lambda_{00} =\left(
\begin{array}{ccccccccc}
 1 & 0 & 0 & 0 & 0 & 0 & 0 & 0 & 0 \\
 0 & \frac{1}{2} & 0 & 0 & 0 & 0 & 0 & 0 & 0 \\
 0 & 0 & 0 & 0 & 0 & 0 & 0 & 0 & 0 \\
 0 & 0 & 0 & 0 & 0 & 0 & 0 & 0 & 0 \\
 0 & 0 & 0 & 0 & 0 & 0 & 0 & 0 & 0 \\
 0 & 0 & 0 & 0 & 0 & 0 & 0 & 0 & 0 \\
 0 & 0 & 0 & 0 & 0 & 0 & 0 & 0 & 0 \\
 0 & 0 & 0 & 0 & 0 & 0 & 0 & 0 & 0 \\
 0 & 0 & 0 & 0 & 0 & 0 & 0 & 0 & 0 \\
\end{array}
\right)
\end{align}
and
\begin{align}
\Lambda_{22} =\left(
\begin{array}{ccccccccc}
 0 & 0 & 0 & 0 & 0 & 0 & 0 & 0 & 0 \\
 0 & 0 & 0 & 0 & 0 & 0 & 0 & 0 & 0 \\
 0 & 0 & 0 & 0 & 0 & 0 & 0 & 0 & 0 \\
 0 & 0 & 0 & 0 & 0 & 0 & 0 & \frac{1}{2} & 0 \\
 0 & 0 & 0 & 0 & 0 & 0 & 0 & 0 & 0 \\
 0 & 0 & 0 & 0 & 0 & 0 & 0 & 0 & 0 \\
 0 & 0 & 0 & 0 & 0 & 0 & 0 & 0 & 0 \\
 0 & 0 & 0 & 0 & 0 & 0 & 0 & 0 & 0 \\
 0 & 0 & 0 & 0 & 0 & 0 & 0 & 0 & 0 \\
\end{array}
\right)  .  
\end{align}

\bibliography{proj}

\end{document}